\documentclass[12pt]{article}
\usepackage{extarrows}
\usepackage{tikzsymbols,tikz,pgfplots}
\usetikzlibrary{snakes}
\usepackage{graphicx}
\usepackage{caption}
\usepackage{subcaption}
\usepackage{xcolor}
\usepackage{algorithm}
\usepackage{algorithmic}

\DeclareMathOperator*{\argmin}{argmin}

\usepackage{bbm}
\usepackage{bm}
\usepackage[toc,page]{appendix}
\usepackage{amsthm}
\usepackage{amssymb}

\RequirePackage[colorlinks,linkcolor=blue,citecolor=blue,urlcolor=blue]{hyperref}
\usepackage{latexsym,epsf,dsfont,color,eurosym,mathrsfs,url}
\usepackage[round,colon,authoryear]{natbib}
\usepackage{amsmath}
\usepackage{mathtools}
\usepackage{graphicx}
\usepackage{float} 
\usepackage{caption}
\usepackage{subcaption}
\usepackage{color,soul}
\usepackage{ulem} 
\usepackage{enumerate}
\usepackage{verbatim}
\usepackage{tabularx}
\usepackage{tabu}

\graphicspath{{figures/}}

\newlength{\fixboxwidth}
\newtheorem{theorem}{Theorem}

\newtheorem{lemma}[theorem]{Lemma}

\theoremstyle{definition}

\theoremstyle{remark}

\numberwithin{equation}{section}

\def\EE{\mathbb E}
\def\II{\mathbb I}
\def\PP{\mathbb P}
\def\RR{\mathbb R}

\def\hat{\widehat}
\def\calP{\mathcal P}
\def\calF{\mathcal F}
\def\rank{\mathrm{rank}}
\def\tr{\mathrm{trace}}

\def\bu{{\mathbf u}}
\def\bv{{\mathbf v}}

\begin{document}

\title{On Recovering the Best Rank-$r$ Approximation from Few Entries$^\ast$}

\date{}

\author{Shun Xu and Ming Yuan$^\dag$\\
Columbia University}

\maketitle

\begin{abstract}
In this note, we investigate how well we can reconstruct the best rank-$r$ approximation of a large matrix from a small number of its entries. We show that even if a data matrix is of full rank and cannot be approximated well by a low-rank matrix, its best low-rank approximations may still be reliably computed or estimated from a small number of its entries. This is especially relevant from a statistical viewpoint: the best low-rank approximations to a data matrix are often of more interest than itself because they capture the more stable and oftentimes more reproducible properties of an otherwise complicated data-generating model. In particular, we investigate two agnostic approaches: the first is based on spectral truncation; and the second is a projected gradient descent based optimization procedure. We argue that, while the first approach is intuitive and reasonably effective, the latter has far superior performance in general. We show that the error depends on how close the matrix is to being of low rank. Both theoretical and numerical evidence is presented to demonstrate the effectiveness of the proposed approaches.
\end{abstract}

\footnotetext[1]{
Research supported in part by NSF Grant DMS-2015285.}
\footnotetext[2]{
Address for Correspondence: Department of Statistics, Columbia University, 1255 Amsterdam Avenue, New York, NY 10027.}

\newpage

\section{Introduction}
Low-rank approximations and other related spectral methods are among the most fundamental and ubiquitous tools in data analysis. Their computational and statistical aspects have been studied extensively and are among the central themes in numerical analysis and multivariate statistics respectively. The classical approaches however face new challenges with massive data sets are being generated every day across diverse fields: gene expression analysis \citep[see, e.g.,][]{kluger2003spectral}, protein-to-protein interaction \citep[see, e.g.,][]{stelzl2005human}, MRI image analysis \citep[see, e.g.,][]{smith2004advances}, and the analysis of large graphs and social networks \citep[see, e.g.,][]{clauset2004finding, scott2017social}, to name a few.  Scalable computation and valid statistical inferences in high dimensions for low-rank approximations are the subjects of fervent research interests in recent years.

Consider a data matrix $A\in \mathbbm{R}^{d_1\times d_2}$. As Eckart-Young Theorem indicates, its best rank-$r$ approximation, denoted by $A_r$ hereafter, can be uniquely identified by its top $r$ singular values and their corresponding vectors, assuming that its $r$th and $(r+1)$th singular values are different for simplicity. From a statistical perspective, we are interested in how well properties of a data-generating model can be inferred from $A_r$, be they principal components, statistical factors, empirical orthogonal functions, or else. Suppose that we are interested in estimating a parameter that can be represented by a matrix $\Theta\in {\mathbb R}^{d_1\times d_2}$ in the Frobenius norm $\|\cdot\|_{\rm F}$. These analyses often tell us about the stochastic behavior of the ``estimation error'' $\|A_r-\Theta\|_{\rm F}$. For example, if $A=\Theta+E$ where $E$'s entries are independent $N(0,\sigma^2)$ variables, then $\|A_r-\Theta\|_{\rm F}^2=O_p(r(d_1+d_2)\sigma^2)$ so that $A_r$ is a consistent estimate of $\Theta$ under the Frobenius norm whenever $\sigma^2\ll [r(d_1+d_2)]^{-1}$. Classical multivariate analysis has primarily focused on ``thin'' data matrices where $d_1$ (sample size) is large and $d_2$ (dimensionality or number of features) is small \citep[see, e.g.,][]{anderson2003introduction}. On the other hand, the emphasis of recent effort is on the so-called high dimensional paradigm where both $d_1$ and $d_2$ are large \citep[see, e.g.,][]{bai2010spectral, wainwright2019high}.

Similarly, from a computational perspective, tremendous progress has been made in the past couple of decades towards fast computation of low-rank approximations to $A$ when both $d_1$ and $d_2$ are large, and randomized algorithms have become increasingly popular for this purpose \citep[see, e.g.,][]{mahoney2011randomized, woodruff2014sketching}. The goal is to find an $\hat{A}$ which may or may not be of rank $r$ that can approximate $A$ nearly as well as $A_r$ under suitable matrix norms. Typical guarantees, in terms of Frobenius norm, for instance, for the output $\hat{A}$ from these algorithms take the form:
\begin{equation}
\label{eq:multerr}
\|A-\hat{A}\|_{\rm F}\le (1+\epsilon)\|A-A_r\|_{\rm F},
\end{equation}
for a sufficiently small factor $\epsilon>0$.

It is, however, worth noting a subtle distinction between these two perspectives. While from a computational point of view, we are interested in finding a ``good'' approximation to $A$; from a statistical viewpoint, the focus is often on $A_r$ itself rather than the original data matrix $A$ because $A_r$ captures the more stable and oftentimes more reproducible properties of an otherwise complicated data-generating model, even if the data matrix $A$ itself may not necessarily be of low rank. These two goals are closely connected yet could also be rather different. In particular, \eqref{eq:multerr} cannot ensure that $\hat{A}$ would inherit the nice statistical properties one may be able to establish for $A_r$: unless $A$ can be well approximated by $A_r$, \eqref{eq:multerr} does not imply that $\|A_r-\hat{A}\|_{\rm F}$ is small so we may not be able to infer from \eqref{eq:multerr} that $\hat{A}$ is necessarily a ``good'' estimate of $\Theta$ even if $\|A_r-\Theta\|_{\rm F}$ is known to be small. Consider, for example, the signal-plus-noise model $A=\Theta+E$ discussed before. Recall that $\|A-A_r\|_{\rm F}^2\asymp_p d_1d_2\sigma^2$ so that $\hat{A}$ may be inconsistent even though $A_r$ is a consistent estimate of $\Theta$ when $(d_1d_2)^{-1}\ll \sigma^2\ll (d_1+d_2)^{-1}$ at least for small $r$s. Moreover, $\hat{A}$ oftentimes has a rank much greater than $r$ and therefore unsuitable for situations where an exact rank-$r$ estimate is sought. Our work aims to fill in this gap between the two strands of fruitful research by investigating how well we can reconstruct $A_r$ from random sparsification of a general matrix $A$, and therefore contributes to growing literature to unify both statistical and computational perspectives \citep[see, e.g.,][]{raskutti2016statistical}.

To facilitate the storage, communication, or manipulation of a large data matrix, one often approximates the original data matrices with a more manageable amount of sketches. See \cite{mahoney2011randomized, woodruff2014sketching} for recent reviews. A popular idea behind many of these approaches is sparsification---creating a sparse matrix by zeroing out some entries of the original data matrix. Sparse sketching of a large data matrix not only reduces space complexity but also allows for efficient computations. See, e.g., \cite{frieze2004fast, arora2006fast, achlioptas2007fast, drineas2008relative, boutsidis2009improved, mahoney2009cur, drineas2011note, achlioptas2013near, krishnamurthy2013low}, among many others. In particular, we shall focus on matrix sparsification schemes that sample each entry of $A$ independently with a prescribed probability.

Denote by $\omega_{ij}$ the indicator that the $(i,j)$ entry of $A$ is sampled and $\Omega=(\omega_{ij})_{1\le i\le d_1,1\le j\le d_2}$. Write $\calP_\Omega(A)=(a_{ij}\omega_{ij}/p_{ij})_{1\le i\le d_1,1\le j\le d_2}$ where $p_{ij}={\mathbb P}(\omega_{ij}=1)$ and we shall adopt the convention that $0/0=0$ in the rest of the paper. Many different sampling schemes have been developed in recent years so that the spectral error $\|\calP_\Omega(A)-A\|$ can be made small even with a tiny fraction of the entries sampled, e.g., $\|\Omega\|_{\ell_1}=\sum_{i,j}|\omega_{ij}|\ll d_1d_2$. See, e.g., \cite{arora2006fast,achlioptas2007fast,drineas2011note} among numerous others. Given the closeness of $\calP_\Omega(A)$ to $A$, it is natural to consider estimating $A_r$ by the best rank-$r$ approximation $[\calP_\Omega(A)]_r$ of $\calP_\Omega(A)$. We show that this is indeed a reasonable approach in light of the tight bounds controlled by $\|\calP_\Omega(A) - A\|$ under the spectral error. But more interestingly, we argue that for an arbitrary $A$, by choosing $p_{ij}$s appropriately, we can derive a much better estimate of $A_r$ for all $r$s. To this end, we introduce a sampling scheme and a companion rank-$r$ estimator $\hat{A}_r$ of $A_r$ from $\calP_\Omega(A)$. We show that $\|\hat{A}_r-A_r\|_{\rm F}$ can be bounded via the sampling error of $\|\calP_\Omega(N_r)-N_r\|$ instead of $\|\calP_\Omega(A)-A\|$ where $N_r=A-A_r$. This leads to a significantly improved estimate of $A_r$. In particular, if $A$ is of rank up to $r$, then $N_r=0$ so that our estimator $\hat{A}_r$ recovers $A$ exactly with high probability. This makes an immediate connection with the popular matrix completion problem where one seeks the exact recovery of a large matrix from partial observations of its entries \citep[see, e.g.,][]{candes2009exact, gross2011recovering, recht2011simpler}

Despite such a similarity, there are also significant distinctions between our results and those typical for matrix completion. For exact matrix completion, it is usually assumed that the rank of $A$ is small and known apriori, and that its singular vectors are incoherent so that each entry carries a similar amount of information. On the other hand, our approach does not make such stipulations and is generally applicable. As we shall show, $\hat{A}_r$ remains a good estimate of $A_r$ even if $A$ is of full rank, and the incoherence condition can be done away with through carefully designed sampling schemes.

The rest of the paper is organized as follows. We first discuss the na\"ive approach to approximate $A_r$ by $[\calP_\Omega(A)]_r$ and show how it connects with existing results for matrix sparsification and completion in the next section. Section 3 introduces our sampling and estimation scheme for the improved approximation of $A_r$. We complement the theoretical results with numerical experiments in Section 4 to further demonstrate the merits of our approach. Proofs of the main results are presented in Section 5 with more technical details relegated to the Appendix.

\section{Na\"ive Estimate of $A_r$}
In the rest of the paper, we shall denote by $A_{i.}$, $A_{.j}$ and $a_{ij}$ the $i$-th row, $j$-th column, and $(i,j)$ entry of a matrix $A$, respectively. Similarly, $v_i$ is the $i$-th element of a vector $v$.  Following the convention, $\|\cdot\|_\textup{F}$ denotes the Frobenius norm, $\|\cdot\|$ the spectral norm for a matrix and the $\ell_2$ norm for a vector, $\|\cdot\|_\infty$ the element-wise infinity norm, $\|\cdot\|_{\ell_1}$ the (vectorized) $\ell_1$ norm. 
 We shall also write the inner product between two matrices as $\langle A,B\rangle=\tr(A^\top B)$. 

Most existing studies of the efficacy of sparsification algorithms focus on bounding the spectral error. For concreteness, consider a scheme that observes $a_{ij}$ independently with probability
\begin{align}\label{p:sparse}
p_{ij}=\min\left\{{n\over 2}\left({a_{ij}^2\over \|A\|_{\rm F}^2}+{|a_{ij}|\over \|A\|_{\ell_1}}\right),1\right\}.
\end{align}
It is not hard to see that, from Chernoff's bound, the total number $\|\Omega\|_{\ell_1}$ of observed entries is within $[n/2, 2n]$ with high probability. Also, there exists a numerical constant $C>4$ such that
$$
\left\|\calP_\Omega(A)-A\right\|\le 4\|A\|_{\rm F}\sqrt{d\over n}+C{\|A\|_{\ell_1}\over n}\sqrt{\log (d)+t},
$$
with probability at least $1-e^{-t}$, where $d=\max\{d_1,d_2\}$ \citep[see, e.g.,][]{o2018random}. In what follows, we shall use $C$, and similarly $C_1$, $C_2$, etc., as generic numerical constants that may take different values at each appearance. Because
$$
\|A\|_{\ell_1}\le d\|A\|_{\rm F},
$$
we have, under the above event,
$$
\left\|\calP_\Omega(A)-A\right\|\le C\|A\|_{\rm F}\left(\sqrt{d\over n}+{d\sqrt{\log (d)+t}\over n}\right).
$$
In particular, whenever $n\ge (\alpha+1)d\log d$, we get
\begin{equation}\label{specerr}
\left\|\calP_\Omega(A)-A\right\|\le C\|A\|_{\rm F}\sqrt{d\over n},
\end{equation}
with probability at least $1- d^{-\alpha}$.
In light of such a bound on the spectral norm, it is natural to consider estimating $A_r$ by the respective best low-rank approximation of $\calP_\Omega(A)$. The connection is made more precise by the following observation that links the closeness between $A_r$ and $[\calP_\Omega(A)]_r$ with the spectral error $\|\calP_\Omega(A)-A\|$.

More specifically, for a matrix $A\in \RR^{d_1\times d_2}$,
its singular values are denoted by
$$
\sigma_1(A)\ge \sigma_2(A)\ge\cdots\ge \sigma_d(A).
$$
Then we have

\begin{lemma}
\label{le:lowrankappr}
Let $B=A+E$. If $\sigma_r(A)>\sigma_{r+1}(A)$ and $\|E\| \leq (\sigma_r(A) - \sigma_{r+1}(A))/2$, then there exists a numerical constant $C_0>0$ such that for any $1\le r<d$,
$$
\|B_r-A_r\|_{\rm F}\le C_0\sqrt{r}\|E\| \left(1+{\sqrt{\sigma_{r+1}(A)\sigma_1(A)}\over \sigma_r(A)-\sigma_{r+1}(A)}\right).
$$
\end{lemma}
In light of Lemma \ref{le:lowrankappr} and \eqref{specerr}, we can bound the difference between $A_r$ and $[\calP_\Omega(A)]_r$ as follows:
\begin{theorem}
\label{th:old}
Assume that each entry of $\Omega$ is independently sampled from binomial trails with probability given by
\begin{equation*}
p_{ij}=\min\left\{{n\over 2}\left({a_{ij}^2\over \|A\|_{\rm F}^2}+{|a_{ij}|\over \|A\|_{\ell_1}}\right),1\right\}.
\end{equation*}
If $\sigma_r(A)>\sigma_{r+1}(A)$, then there exist numerical constants $C_1,\ C_2>0$ such that for any $\alpha>0$,
$$
\|[\calP_\Omega(A)]_r-A_r\|_{\rm F}\le C_1\|A\|_{\rm F}\sqrt{rd\over n}\left(1+{\sqrt{\sigma_{r+1}(A)\sigma_1(A)}\over \sigma_r(A)-\sigma_{r+1}(A)}\right),
$$
with probability at least $1-d^{-\alpha}$ provided that
$$
n\ge C_2d\max\left\{(1+\alpha)\log d,{\|A\|_{\rm F}^2\over (\sigma_{r}(A) - \sigma_{r+1}(A))^2}\right\}.
$$
\end{theorem}

Theorem \ref{th:old} justifies $[\calP_\Omega(A)]_r$ as a valid estimator of $A_r$. It is worth noting that both the sample size requirement and error bound depend on the eigengap $\sigma_r(A)-\sigma_{r+1}(A)$. This is inevitable since the eigengap characterizes the stability of best rank-$r$ approximation and in the extreme case where eigengap is 0, $A_r$ is not uniquely defined. We hereafter assume that $\sigma_r(A) > \sigma_{r+1}(A)$.

It is instructive to revisit the signal-plus-noise model we discussed in the Introduction: $A=\Theta+E\in \RR^{d\times d}$ where, to fix ideas, we assume that the signal $\Theta=\bu\bv^\top$ is a rank-one square matrix with both $\bu$ and $\bv$ of unit length, and that the noise $E$ has independent $N(0,\sigma^2)$ entries. As noted before, it is of special interest to focus on the case when $d^{-2}\ll\sigma^2\ll d^{-1}$ and there is a mismatch between the existing results from the statistical and computational sides. It is not hard to see that in this case $\|A\|_{\rm F}=O_p(d\sigma)$ and $\|A_1-\Theta\|_{\rm F}=O_p(\sigma\sqrt{d})$. In light of Theorem \ref{th:old}, this suggests that
$$
\|[\calP_\Omega(A)]_1-\Theta\|_{\rm F}\le \|[\calP_\Omega(A)]_1-A_1\|_{\rm F}+\|A_1-\Theta\|_{\rm F}=O_p\left(\sigma d^{3/2} n^{-1/2}\right).
$$
Hence $[\calP_\Omega(A)]_1$ is consistent under the sample size requirement that $n\ge C\sigma^2d^3\log d$. In particular, if $\sigma^2\sim d^{-\alpha}$ for some $\alpha\in (1,2)$, then the sample size requirement can be expressed as $n\gtrsim d^{3-\alpha}\log d$ which means that a consistent estimate of $\Theta$ can be obtained from a vanishing proportion ($n/d^2\to 0$) of the entries of $A$.

Another interesting test case here is when $A$ is of rank up to $r$ so that $A=A_r$. If, in addition, all entries of $A$ are of the same order so that $p_{ij}\sim n/(d_1d_2)$, then it is not hard to see that
$$
\EE\|\calP_\Omega(A)-A\|_{\rm F}^2\le {C\|A\|_{\rm F}^2d^2\over n},
$$
so that we have the following performance bound for $\calP_\Omega(A)$ as an approximation to $A$ under the Frobenius norm:
$$
\|\calP_\Omega(A)-A\|_{\rm F}=O_p\left(\|A\|_{\rm F}\sqrt{d^2\over n}\right).
$$
On the other hand, Theorem \ref{th:old} indicates that
$$
\|[\calP_\Omega(A)]_r-A\|_{\rm F} = O_p\left(\|A\|_{\rm F}\sqrt{rd\over n}\right).
$$
This immediately suggests that $[\calP_\Omega(A)]_r$ is a better estimate of $A$ by leveraging the fact that $\rank(A)$ is small. Although this example shows the efficacy of $[\calP_\Omega(A)]_r$ as an estimate of $A_r$, it also points to room for further improvement at least when $A$ is of low rank. In fact, with additional conditions, it may even be possible to recover a low-rank matrix $A$ exactly from $\calP_\Omega(A)$!

For this to be possible, one usually assumes that $\rank(A)$ is small and known apriori. Another essential concept is the so-called incoherence condition. Denote by $A = U\Sigma V^\top$ its singular value decomposition. We say a rank-$r$ matrix $A$ is $\mu$-incoherent if
$$\max\left\{{d_1\over r} \|U_{i.}\|^2, {d_2\over r} \|V_{j.}\|^2\right\}\le \mu.$$
Intuitively, the incoherence condition ensures that each entry of $A$ is of similar importance so that missing any one of them will not prevent us from being able to recover $A$.  Numerous tractable algorithms have been developed to reconstruct $A$ assuming that it is $\mu$-incoherent and each of its entries is sampled independently with probability $n/(d_1d_2)$. See, e.g., \cite{candes2009exact,candes2010power,keshavan2010matrix,sun2016guaranteed} among many others. Interested readers are also referred to \cite{davenport2016overview, chen2018harnessing} for a couple of recent surveys. For example, a natural approach is to reconstruct $A$ by the solution to
$$\min_{\rank(B)\le r} \left\langle \calP_\Omega\left(A-B\right),A-B \right\rangle.$$
In particular, results from the recent work of \cite{chen2020noisy} indicate that we can recover $A$ exactly this way, with high probability, if
$$n \geq C\kappa^4\mu^2r^2(d_1+d_2)\log^3 (d_1+d_2),$$
where $\kappa=\sigma_1(A)/\sigma_r(A)$ is the condition number of $A$.

These results immediately suggest that we can do even better than $[\calP_\Omega(A)]_r$, albeit under additional assumptions. But can we do better than $[\calP_\Omega(A)]_r$ in general? Especially, can we do away with the incoherence assumption and the need to know $\rank(A)$ apriori? We shall now argue that the answer is indeed affirmative.

\section{Improved Estimate of $A_r$}

There are two main ingredients to our approach: sparsification with carefully chosen sampling probabilities to remove the need for incoherence; and an agnostic procedure based on projected gradient descent to estimate $A_r$ for any $r$. 

The fact that incoherence is closely related to the uniform sampling was noted first by \citet{chen2015completing}. They observed that a rank-$r$ matrix $A$ could be recovered exactly without the incoherence condition by taking 
$$p_{ij} \propto {d_1\over r} \|U_{i.}\|^2+ {d_2\over r} \|V_{j.}\|^2.$$
The difficulty of course is that it requires that $A$ be of rank $r$. Nonetheless, motivated by the observation that
$${\|A_{i.}\|\over\sigma_1(A)} \le \|U_{i.}\| \le {\|A_{i.}\|\over\sigma_r(A)},$$
and
$${\|A_{.j}\|\over\sigma_1(A)} \le \|V_{j.}\| \le {\|A_{j.}\|\over\sigma_r(A)},$$
we shall consider sampling $a_{ij}$ with probability
\begin{equation}
\label{eq:defp}
p_{ij}=\min\left\{{n\over 3}\left({\|A_{i\cdot}\|^2\over d_2\|A\|_{\rm F}^2}+{\|A_{\cdot j }\|^2\over d_1\|A\|_{\rm F}^2}+{|a_{ij}|\over \|A\|_{\ell_1}}\right), 1\right\},
\end{equation}
Compared with \eqref{p:sparse}, we essentially replaced $a_{ij}^2$ with $\|A_{i.}\|^2+\|A_{.j}\|^2$ which plays a similar role as $(d_1\|U_{i.}\|^2+d_2\|V_{j.}\|^2)/r$ when $A$ is indeed of rank $r$. Unlike the factor $(d_1\|U_{i.}\|^2+d_2\|V_{j.}\|^2)/r$, however, our choice of $p_{ij}$ does not depend on $r$. This is critical in allowing us to estimate $A_r$ for any $r$ from the sampled entries.

We shall now describe how to reconstruct $A_r$ from $\calP_\Omega (A)$. Our approach is similar in spirit to \cite{ge2017no} and \cite{chen2020noisy}. See also \cite{jain2017non}. Recall that $A_r$ is the solution to
$$
\argmin_{\rank(B)\le r}\|A-B\|_{\rm F}^2.
$$
It is clear by construction
$$\EE \left\langle \calP_\Omega\left(A-B\right),A-B \right\rangle=A-B,$$
so that we can consider minimizing $\left\langle \calP_\Omega\left(A-B\right),A-B \right\rangle$. Any matrix $B$ with rank up to $r$ can be written as $B=XY^\top$ for some $X\in \RR^{d_1\times r}$ and $Y\in \RR^{d_2\times r}$ so that this is equivalent to minimizing $\left\langle \calP_\Omega\left(A-XY^\top\right),A-XY^\top \right\rangle$ over the couple $(X,Y)$. However, decomposition $B=XY^\top$ is not unique. To overcome such an identifiable issue, we shall consider estimating $A_r$ by the solution to
\begin{equation}\label{defopt}
\min_{\calF_\beta} f(X,Y) =\frac{1}{2} \left\langle \calP_\Omega\left(A-XY^\top\right),A-XY^\top \right\rangle + \frac{1}{8} \left\|X^\top X - Y^\top Y\right\|_{\rm F}^2 ,
\end{equation}
where 
$$\calF_\beta=\left\{ (X,Y)\in  \RR^{d_1 \times r}\times \RR^{d_2\times r}: \|X_{i.}\|\leq \|A_{i.}\|/\beta,\ \|Y_{j.}\|\leq \|A_{.j}\|/\beta \right\},$$ 
and $\beta$ is a tuning parameter to be specified later. 

The second term on the right-hand side  of \eqref{defopt} forces $X^\top X=Y^\top Y$ so that if $XY^\top =\tilde{U}\tilde{\Sigma}\tilde{V}^\top$ then $X=\tilde{U}\tilde{\Sigma}^{1/2}$ and $Y= \tilde{V}\tilde{\Sigma}^{1/2}$. Note that
$$\|(A_r)_{i.}\|^2  = \sum_{j=1}^r \sigma_j(A)^2u_{ij}^2 \geq \sigma_r(A) \sum_{j=1}^r \sigma_j(A) u_{ij}^2.$$
When $X$ to be close to $U_r\Sigma^{1/2}_r$, then
$$\|X_{i.}\|^2 \approx \sum_{j=1}^r\sigma_j(A) u_{ij}^2 \leq \|(A_r)_{i.}\|^2/\sigma_r(A) \le \|A_{i.}\|^2/\sigma_r(A),$$
so that the constraint on $X_{i\cdot}$, and similarly that on $Y_{j\cdot}$, becomes inactive at least for a sufficiently small $\beta$.

The objective function $f$ in \eqref{defopt} is non-convex jointly over $(X,Y)$. Nonetheless, it is natural to solve it by projected gradient descent with a good initialization. To this end, let $\calP_{\calF_\beta}$ be the projection on the convex set $\calF_\beta$, that is
$$
\calP_{\calF_\beta}(X)_{i.} = \frac{X_{i.}}{\|X_{i.}\|} \cdot \min\left\{\|X_{i.}\|,\frac{\|A_{i.}\|}{\beta} \right\},
$$
and
$$
\calP_{\calF_\beta}(Y)_{j.} = \frac{Y_{j.}}{\|Y_{j.}\|} \cdot \min\left\{\|Y_{j.}\|,\frac{\|A_{.j}\|}{\beta} \right\}.
$$

\begin{algorithm}[H]
\begin{algorithmic}[1]
\caption{Projected Gradient Descent}\label{alg1}
\REQUIRE Step size $\eta$, number of iterations $T$, tuning parameter $\beta$
\STATE  \textbf{Initialization}: $X_0=U_0\Sigma_0^{1/2}$, $Y_0=V_0\Sigma_0^{1/2}$, where $[\calP_\Omega(A)]_r =U_0\Sigma_0V_0^\top$ is its SVD
\STATE \textbf{for}: $t=0,1,2,\dots,T-1$
\STATE $\qquad X_{t+1} = X_t-\eta\nabla_X f(X_t,Y_t)$, $Y_{t+1} = Y_t-\eta\nabla_Y f(X_t,Y_t)$
\STATE $\qquad X_{t+1} = \calP_{\calF_\beta}(X_{t+1})$, $Y_{t+1} = \calP_{\calF_\beta}(Y_{t+1})$
\STATE \textbf{end for}
\ENSURE $(X_T,Y_T)$
\end{algorithmic}
\end{algorithm}

The intuition behind our approach is as follows. Write $A=A_r+N_r$. In light of the discussion from the previous subsection, we may be able to exactly recover $A_r$ from $\calP_\Omega(A_r)$. However, what we have is $\calP_\Omega(A) = \calP_\Omega(A_r) + \calP_\Omega(N_r)$, with an extra term $\calP_\Omega(N_r)$, and we now hope to be able to control the error bound for $\hat{A}_r:=X_TY_T^{\top}$ by the spectral error $\left\|\calP_\Omega (N_r)-N_r\right\|$. The main advantage of our scheme is that we have tighter control of the perturbation $\left\|\calP_\Omega(N_r)-N_r\right\|$ than $\left\|\calP_\Omega(A)-A\right\|$ used in Theorem \ref{th:old}. Denote by
$$
\mu_r(A) = \max_{\substack{1\le i \le d_1\\1\le j \le d_2}}\left\{{\|(A_r)_{i\cdot}\|\over \|A_{i\cdot}\|},{\|(A_r)_{\cdot j}\|\over \|A_{\cdot j}\|}\right\},
$$
$$
\nu_r(A)^2=1-\min_{\substack{1\le i \le d_1\\1\le j \le d_2}}\left\{{\|(A_r)_{i\cdot}\|^2\over \|A_{i\cdot}\|^2},{\|(A_r)_{\cdot j}\|^2\over \|A_{\cdot j}\|^2}\right\},
$$
and
$$
\nu_r(A)_\infty = \max_{\substack{1\le i \le d_1\\1\le j \le d_2}}\left\{{|a_{ij}-(a_r)_{ij}|\over |a_{ij}| + (\|A_{i \cdot}\|^2 + \|A_{\cdot j}\|^2)/\|A\|_{\rm F}}\right\}.$$
We have
\begin{lemma}
\label{le:remainder}
Assume that each entry of $\Omega$ is independently sampled from binomial trails with probability given by \eqref{eq:defp}. If $\sigma_r(A)>\sigma_{r+1}(A)$, then there exist numerical constants $C_0,C_1>0$, with probability at least $1-d^{-\alpha}$, such that
\begin{equation}
\label{remainder:bound1}
\left\|\calP_\Omega(N_r)-N_r\right\|\le C_0\left(\nu_r(A)+ \nu_r(A)_\infty\sqrt{d\log d\over n}\right)\|A\|_{\rm F}\sqrt{d\over n}.
\end{equation}
Further, an upper bound of $\nu_r(A)_\infty$ gives us 
\begin{equation}
\label{remainder:bound2}
\left\|\calP_\Omega(N_r)-N_r\right\|\le C_1\max \left\{ \nu_r(A),\left(1+{\mu_r(A)^2\|A\|_{\rm F}\over\sigma_r(A)}\right)\sqrt{d\log d\over n}\right\}\|A\|_{\rm F}\sqrt{d\over n}.
\end{equation}
\end{lemma}
It is worth pointing out that the sampling probability, and consequently $\Omega$ is determined by $A$ instead of $N_r$. As a result, $\calP_\Omega(N_r)-N_r$ is not determined by $N_r$ alone and the error bound above depends also on properties of $A$. Compared with \eqref{specerr}, for sufficiently large $n$, $\nu_r(A)$ describes how much tighter control we can have for $\left\|\calP_\Omega(N_r)-N_r\right\|$ than $\left\|\calP_\Omega(A)-A\right\|$. By definition, $ \nu_r(A) \leq 1$, this means $\calP_\Omega(N_r)$ concentrates around its mean more tightly than $\calP_\Omega(A)$. We are now in the position to present our main theoretical guarantee which shows that the quality of recovery indeed rests upon the perturbation of $\calP_\Omega(N_r)$, not $\calP_\Omega(A)$. 

To this end, write $X^* = U_r\Sigma_r^{1/2}, Y^* = V_r\Sigma^{1/2}_r$ and denote by 
$$F_t = \begin{bmatrix}
X_t \\Y_t
\end{bmatrix},\quad F^* = \begin{bmatrix}
X^* \\Y^*
\end{bmatrix},$$ 
where $A_r= U_r\Sigma_r V_r^\top$ is its SVD. Due to rotation symmetry, the difference between $F_t$ ad $F^\ast$ can be measured by
$$\delta_t = \min_{RR^\top = R^\top R = I} \left\|F_t - F^*R \right\|_\textup{F}.$$
We have

\begin{theorem}
\label{th:main}
Assume $\sigma_r(A)>\sigma_{r+1}(A)$ and each entry of $\Omega$ is independently sampled with probability given by \eqref{eq:defp}, denote by $(X_t,Y_t)$ the iteration sequence from Algorithm \ref{alg1} with $\eta \leq (\sigma_r(A)-\sigma_{r+1}(A))\beta^4(1-\nu_r(A)^2)^2/(500\|A_r\|_{\rm F}^4)$. Then there exist numerical constants $C_0, C_1>0$ such that
$$\delta_{t}   \leq \left[ 1-\frac{1}{5}\eta\left(\sigma_r(A)-\sigma_{r+1}(A)\right) \right]^{t/2}\cdot \delta_{0} +
 C_0 \sqrt{r}\left\|\calP_\Omega(N_r)-N_r\right\| \cdot \frac{\sqrt{\sigma_r(A)}}{\sigma_r(A)-\sigma_{r+1}(A)},$$
with probability at least $1-d^{-\alpha}$ provided that $\beta \leq \sqrt{\sigma_r(A)/(1-\nu_r(A)^2)}$ and
\begin{eqnarray*}
n&\ge& C_1(1+\alpha)d\log d \cdot {\|A\|_{\rm F}^2\over(\sigma_r(A)-\sigma_{r+1}(A))^4}\\
&&\cdot \max\left\{ \nu_r(A)^2r\sigma_1(A)\sigma_r(A),{r\sigma_1(A)^2\sigma_{r+1}(A)\over\sigma_r(A)},{\sigma_r^2(A)\over \beta^4(1-\nu_r(A))^2} \mu_r^2(A) \|A_r\|_{\rm F}^2\right\},
\end{eqnarray*}
where $d=\max\{d_1,d_2\}$. Moreover,
\begin{eqnarray*}
\|X_tY_t^\top-A_r\|_{\rm F}\le \sqrt{3\sigma_1(A)}\left[ 1-\frac{1}{5}\eta\left(\sigma_r(A)-\sigma_{r+1}(A)\right) \right]^{t/2}\cdot \delta_{0}\\
+C_0 \sqrt{r}\left\|\calP_\Omega(N_r)-N_r\right\|\cdot \frac{\sigma_r(A)}{\sigma_r(A)-\sigma_{r+1}(A)}.
\end{eqnarray*}
\end{theorem}

A few remarks follow immediately. A benchmark case is when $A$ is of low rank. More specifically if $\rank(A) \le r$, then $A_r = A$ and $\nu_r(A) =\nu_r(A)_\infty =0$. From Theorem \ref{th:main} and \eqref{remainder:bound1}, 
$$\delta_{t}   \leq \left[ 1-\frac{1}{5}\eta\sigma_r(A) \right]^{t/2}\cdot \delta_{0},$$
which implies that $F_t\to F^\ast$ and consequently $X_t{Y_t}^\top \to  A$ as $t\to \infty$. It is worth pointing out that if $\rank(A) < r$, the $\sigma_r(A)$ in the above inequality can be replaced by $\sigma_{\rank(A)}(A)$. In other words, the exact recovery of $A$ can be achieved without either the incoherence condition or knowing the precise value of $\rank(A)$ apriori. Furthermore, if the incoherence condition indeed holds, the sample size requirement from Theorem \ref{th:main} is comparable to those typical in the matrix completion literature.

In general, Theorem \ref{th:main} and \eqref{remainder:bound2} suggests that whenever $\sigma_r(A)>\sigma_{r+1}(A)$,
\begin{eqnarray}
\label{newerr}
\|X_tY_t^\top-A_r\|_{\rm F}&\le& C_0\|A\|_{\rm F}\sqrt{rd\over n}\cdot \frac{\sigma_r(A)}{\sigma_r(A)-\sigma_{r+1}(A)} \nonumber\\
&& \cdot\max \left\{ \nu_r(A),\left(1+{\mu_r(A)^2\|A\|_{\rm F}\over\sigma_r(A)}\right)\sqrt{d\log d\over n}\right\},
\end{eqnarray}
for large enough $t$. This is to be compared with the bound given by Theorem \ref{th:old}:
$$
\|[\calP_\Omega(A)]_r-A_r\|_{\rm F}\le C_0\|A\|_{\rm F}\sqrt{rd\over n}\left(1+{\sqrt{\sigma_{r+1}(A)\sigma_1(A)}\over \sigma_r(A)-\sigma_{r+1}(A)}\right).
$$
It is clear that the bound in \eqref{newerr} is smaller because
$$\nu_r(A)\le1,\quad \mu_r(A)\le1,\quad n \gtrsim d\log d\|A\|_{\rm F}^2/\sigma_r(A)^2,$$ 
and
\begin{eqnarray*}
1+{\sqrt{\sigma_{r+1}(A)\sigma_1(A)}\over \sigma_r(A)-\sigma_{r+1}(A)}&=&\frac{\sigma_r(A)}{\sigma_r(A)-\sigma_{r+1}(A)}+{\sqrt{\sigma_{r+1}(A)\sigma_1(A)}-\sigma_{r+1}(A)\over \sigma_r(A)-\sigma_{r+1}(A)}\\
&>&\frac{\sigma_r(A)}{\sigma_r(A)-\sigma_{r+1}(A)}.
\end{eqnarray*}
To gain further insights of the difference between the two estimates of $A_r$, we shall now take a closer look at these factors.

For brevity, we shall assume $d_1=d_2=d$ for the discussion. Recall that $A=UDV^\top$. It is not hard to see that
$$\nu_r(A)^2 = \max_{1\le j\le d}  \left\{ \frac{\sum_{i=r+1}^d\sigma_i^2 u_{ij}^2}{\sum_{i=1}^d\sigma_i^2 u_{ij}^2},  \frac{\sum_{i=r+1}^d\sigma_i^2 v_{ij}^2}{ \sum_{i=1}^d\sigma_i^2 v_{ij}^2} \right\}.$$
If the entries of $U$ and $V$ are of similar order, then 
$$\nu_r(A)^2\approx \sum_{i=r+1}^d\sigma_i^2/\sum_{i=1}^d\sigma_i^2=\|A-A_r\|_{\rm F}^2/\|A\|^2_{\rm F}.$$
Similarly, $\mu_r(A) \approx d\|A-A_r\|_\infty/\|A\|_{\rm F}$ so that, from \eqref{newerr}, there is more improvement over the na\"ive estimate $[\calP_\Omega(A)]_r$ when $A$ is close to being of rank $r$.

On the other hand, it is instructive to consider the case when $\sigma_i(A)= i^{-\alpha}$. In this case,
$$
\frac{\sigma_r(A)}{\sigma_r(A)-\sigma_{r+1}(A)}\sim r,
$$
and on the other hand,
$$
1+{\sqrt{\sigma_{r+1}(A)\sigma_1(A)}\over \sigma_r(A)-\sigma_{r+1}(A)}\sim r^{\alpha/2+1}.
$$
This implies that much better bound can be achieved by our approach as $r$ and $\alpha$ increase.

\section{Numerical Experiments}

To complement the theoretical developments from the previous sections and further demonstrate the merits of our approach, we also conducted several sets of numerical experiments. In particular, we focused on several key operating characteristics of our method including the ability to recover low-rank matrices exactly; the impact of the target rank $r$; and the role of the eigengap $\sigma_r(A)-\sigma_{r+1}(A)$. We consider estimating $A_r$ by the na\"ive estimate $[\calP_\Omega(A)]_r$ with $\Omega$ generated according to \eqref{p:sparse}, and by $\hat{A}_r:=X_TY^\top_T$ as the output from Algorithm 1 with $\Omega$ generated according to \eqref{eq:defp}. For a fair comparison, we adjusted $n$ of the sampling probability \eqref{p:sparse} (for na\"ive method) and \eqref{eq:defp} (for the improved estimate) to ensure that both methods sample the same number of entries on average. In our implementation of the projected gradient descent,  $\beta$ was set to be $\sqrt{\sigma_r(\calP_\Omega(A))/2}$. To accelerate the optimization, we used line search (the line\_search function of scipy.optimize in Python) to determine the optimal stepsize. The stopping criterion was set to be $\|\nabla f(X,Y)\|_\infty \leq 10^{-6}$. In what follows, we shall refer to the ``relative error'' as $\|\hat{A}_r-A_r\|_\textup{F}/\|A_r\|_\textup{F}$ or $\|[\calP_\Omega(A)]_r-A_r\|_\textup{F}/\|A_r\|_\textup{F}$, and the ``spectral error'' as $\|\hat{A}_r-A_r\|/\|A_r\|$ or $\|[\calP_\Omega(A)]_r-A_r\|/\|A_r\|$. All the results presented are based on 100 independent simulation runs.

In our first set of experiments, we set $d_1 = d_2 =1000$ and generated $A$ as a low-rank matrix $A^*$ with additive noise $\sigma E$. More specifically, 
$$A^* = U^*\Sigma^* {U^*}^\top/\|U^*\|_\textup{F}^2,\qquad \Sigma^* = \text{diag}\{1,0.9,0.8,0.7,0.6\},$$ 
where $U^\ast\in \RR^{1000\times 5}$ is a Gaussian ensemble with entries independently sampled from the standard normal distribution, and $E$'s entries are independently from $N(0,1/4000)$. The variance of $E_{ij}$ is chosen such that $\|E\| \approx 1$.

We began with the case where $A$ is of low rank. More specifically, we set $A = A^*$ and $r=5$. In Figure \ref{fig:low rank}, we plotted the relative error of $[\calP_\Omega(A)]_r$ and $\hat{A}_r$ with the max number of iterations set to be $T=5$ or $10$. It is clear that the relative error of both estimates decreases quickly with the increased sample size, as predicted by our theoretical results. But the relative error of $\hat{A}_r$ decreases quickly to 0 while the error of the na\"ive estimate levels off around 0.15 even with as many as 50\% of the entries observed. This highlights the ability of exact recovery for the improved approach as discussed in the previous section. Empirically, the projected gradient descent algorithm converges fairly quickly: there is little difference between setting the max number of iterations to 5 and 10, and 10 seems to suffice as a rule of thumb. In the rest of the experiments, we shall fix the max number of iterations at ten for consideration of computational speed. 

\begin{figure}
\centering
\includegraphics[width=0.6\textwidth]{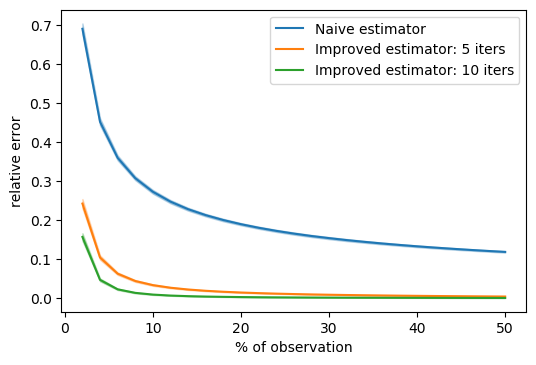}
\caption{``Relative error'' of different methods under low-rank scenario, averaged over 100 simulation runs. }
\label{fig:low rank}
\end{figure}

Next, we consider the case where $A$ is full rank yet we are primarily interested in its best rank approximations. To this end, we set $A = A^* + 0.05E$. We adjusted $n$ of the sampling probability so that on average 10\% of the entries were sampled. Figure \ref{fig: target rank} reports the mean and two-standard-deviance bands of both relative spectral and Frobenius errors from 100 simulation runs. A few interesting observations can be made. In particular, it is evident that $\hat{A}_r$ is superior to the na\"ive estimator, in either error metric. Moreover, the na\"ive estimate is much more vulnerable to overshooting the ``effective'' rank of $A$. Note that $A$ is ``close'' to being of rank $5$. When consider estimating $A_r$ for $r>5$, the performance of $\hat{A}_r$ only deteriorates mildly with an increasing $r$ yet on the other hand, for the na\"ive estimate the impact is much more significant.

\captionsetup[subfigure]{labelfont=rm}
\begin{figure}
\centering
\begin{subfigure}{0.45\textwidth}
\includegraphics[width=\textwidth]{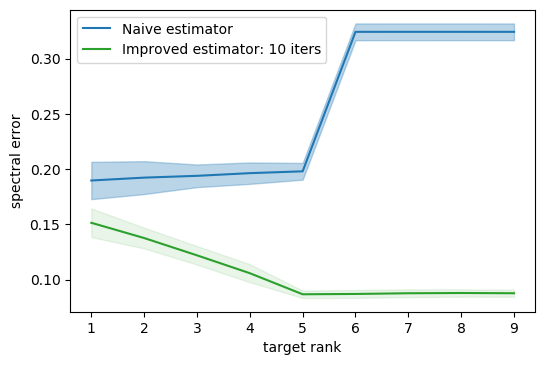}
\end{subfigure}
\begin{subfigure}{0.45\textwidth}
\includegraphics[width=\textwidth]{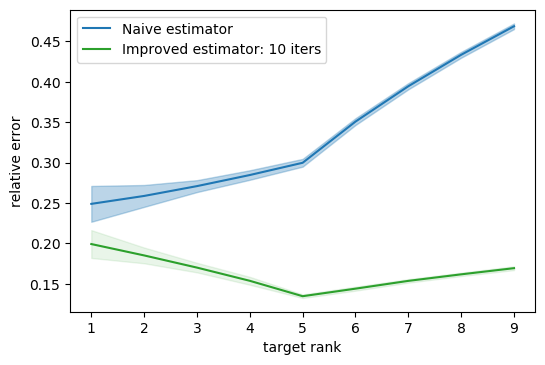}
\end{subfigure}
\caption{Effect of the targeted rank $r$ on different methods: averaged errors and $\pm$ two-standard-error bands were based on 100 simulation runs.}
\label{fig: target rank}
\end{figure}

To further investigate the impact of the eigengap $\sigma_r(A)-\sigma_{r+1}(A)$, we considered $A = A^*+ \sigma E$ with $\sigma$ varying from 0.02 to 0.4. Here $\sigma$ serves as a proxy of the relative eigengap $1-\sigma_6(A)/\sigma_5(A)$ as $\|E\| \approx 1$. The results, again based upon 100 simulation runs, were summarized in Figure \ref{fig:eigengap}. It is interesting to note that the error increases as the relative eigengap decreases for both methods, but the impact on $\hat{A}_r$ is minimal when compared with the na\"ive method, especially with increased sampling proportions.
\captionsetup[subfigure]{labelfont=rm}
\begin{figure}
\centering
\begin{subfigure}{0.45\textwidth}
\includegraphics[width=\textwidth]{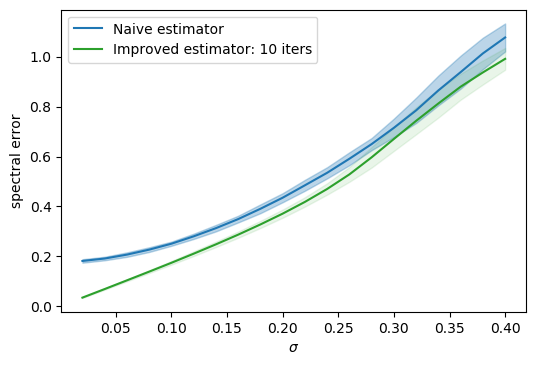}
\caption{10\% observations}
\end{subfigure}
\begin{subfigure}{0.45\textwidth}
\includegraphics[width=\textwidth]{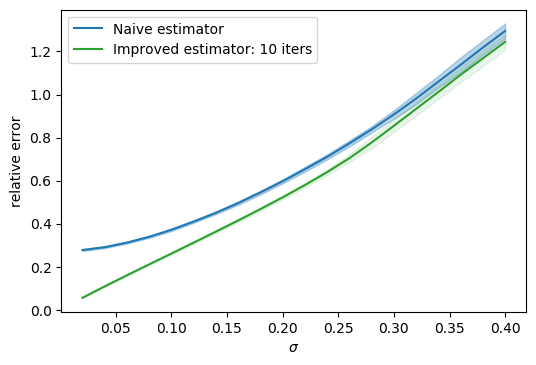}
\caption{10\% observations}   
\end{subfigure}
\vfill
\begin{subfigure}{0.45\textwidth}
\includegraphics[width=\textwidth]{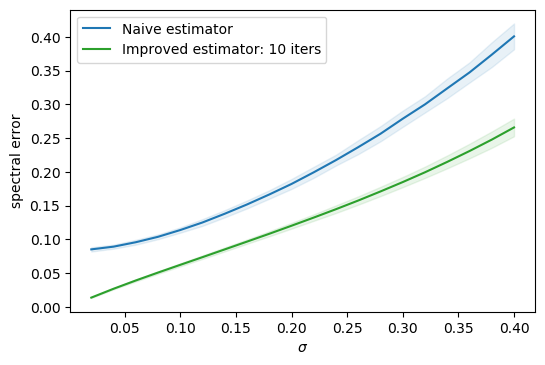}
\caption{40\% observations}
\end{subfigure}
\begin{subfigure}{0.45\textwidth}
\includegraphics[width=\textwidth]{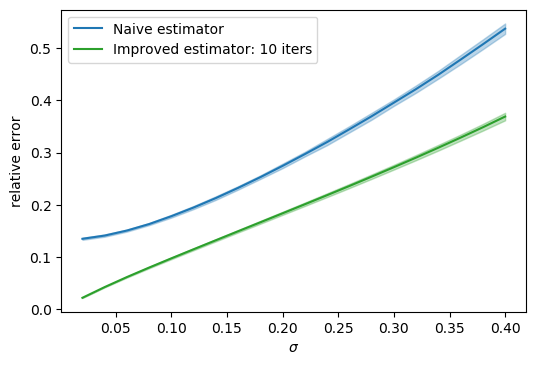}
\caption{40\% observations}   
\end{subfigure}
\caption{Effect of eigengap and sampling proportion on different methods: averaged errors and $\pm$ two-standard-error bands were based on 100 simulation runs.}
\label{fig:eigengap}
\end{figure}

In our final example, we computed both estimators to the waterfall image \footnote{The original color version of image can be downloaded at \url{https://media.cntraveler.com/photos/571945e380cf3e034f974b7d/master/pass/waterfalls-Seljalandsfoss-GettyImages-457381095.jpg}.}. The original image was converted into greyscale. 
The leading singular value of the $1536 \times2056$ matrix consisting of the pixel intensities accounts for 57.5\% of the total variation. The best rank-10, rank-20, and rank-30 approximations explain 88.1\%, 91.6\%, and 93.0\% of the variation respectively. We considered estimating the best rank-$r$ approximation to the image with $r=5, 10$, $20$, and $30$. For each $r$, we set the max number of iteration to be 10 and sampled 10\% pixels. Table \ref{table:waterfall} reports the mean and standard deviation of the ``relative error'' from 100 simulation runs. They again confirm that $\hat{A}_r$ is a far more accurate estimate of $A_r$.

\begin{table}
\centering
\begin{tabular}{|c|c|c|c|c|}
\hline
 Target rank $r$& Na\"ive Estimate &Improved Estimate  \\
\hline
5 & 0.290 (0.0017)& 0.118 (0.0036) \\
10 & 0.432 (0.0012)& 0.154 (0.0018)\\
20& 0.609 (0.0009)& 0.174 (0.0012)\\
30&0.738 (0.0009)& 0.184 (0.0010)\\
\hline
\end{tabular}
\caption{Comparison between the na\"ive and the improved estimates of the best rank-$r$ approximations of the greyscale waterfall image for different targeted ranks. Reported are averaged ``relative error'' over 100 runs and numbers in parentheses are standard deviations.}
\label{table:waterfall}
\end{table}
To facilitate visual comparison between the improved and na\"ive methods, we now focus on the best rank-$30$ approximation of the image as shown in panel (a) of Figure \ref{fig:waterfall}. We fixed $n$ in the sampling probabilities so that the expected sampling proportion is 40\%. One typical realization of both estimates is given in Panels (b) and (c) of Figure \ref{fig:waterfall}. For this specific realization, $\hat{A}_r$ has a relative error 0.11. This is to be compared with the na\"ive method which has a relative error 0.37, which again is in agreement with our theoretical findings.

\captionsetup[subfigure]{labelfont=rm}
\begin{figure}
\centering
\begin{subfigure}{0.3\textwidth}
\includegraphics[width=\textwidth]{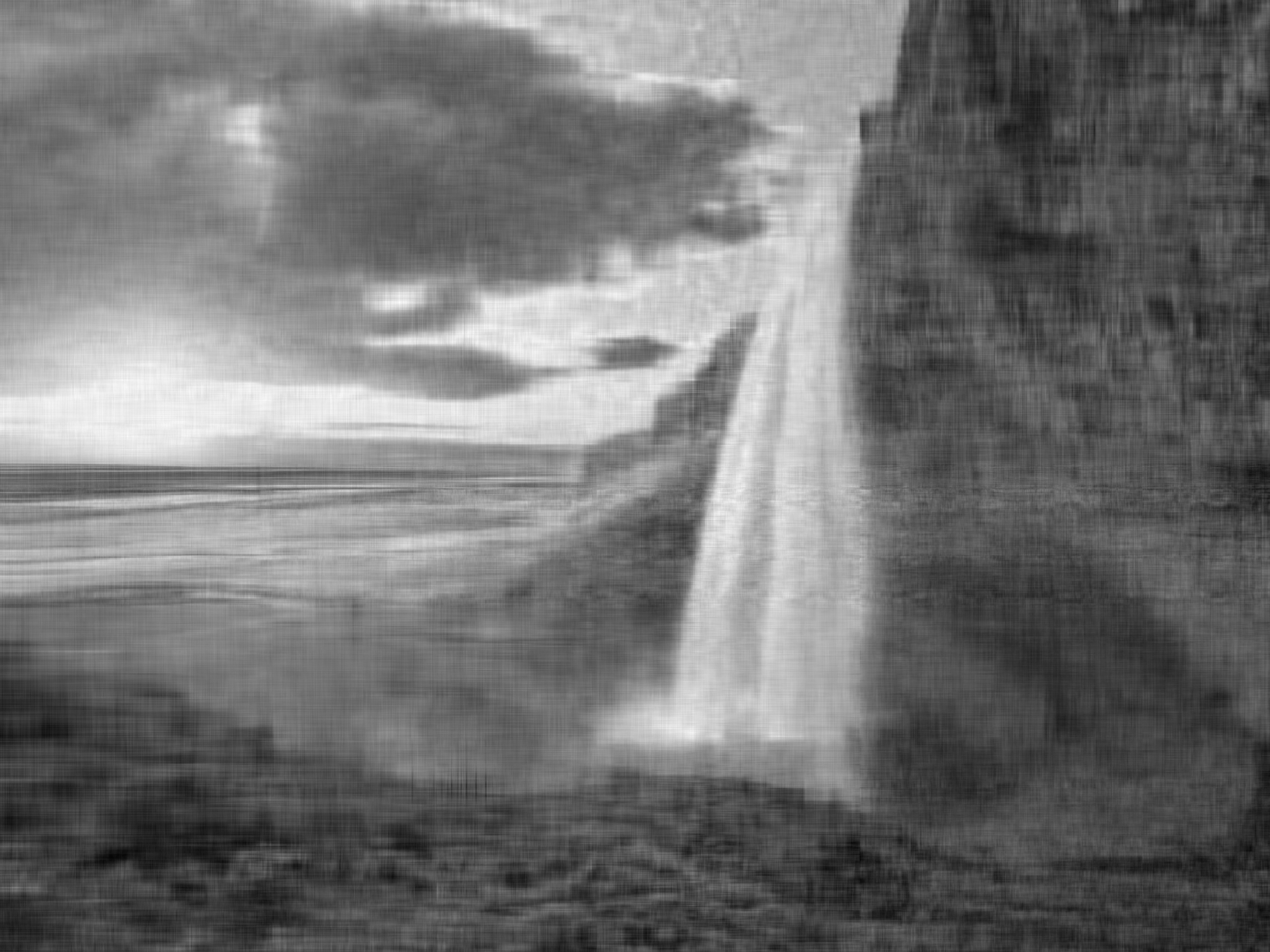}
\caption{Exact approximation}   
\end{subfigure}
\begin{subfigure}{0.3\textwidth}
\includegraphics[width=\textwidth]{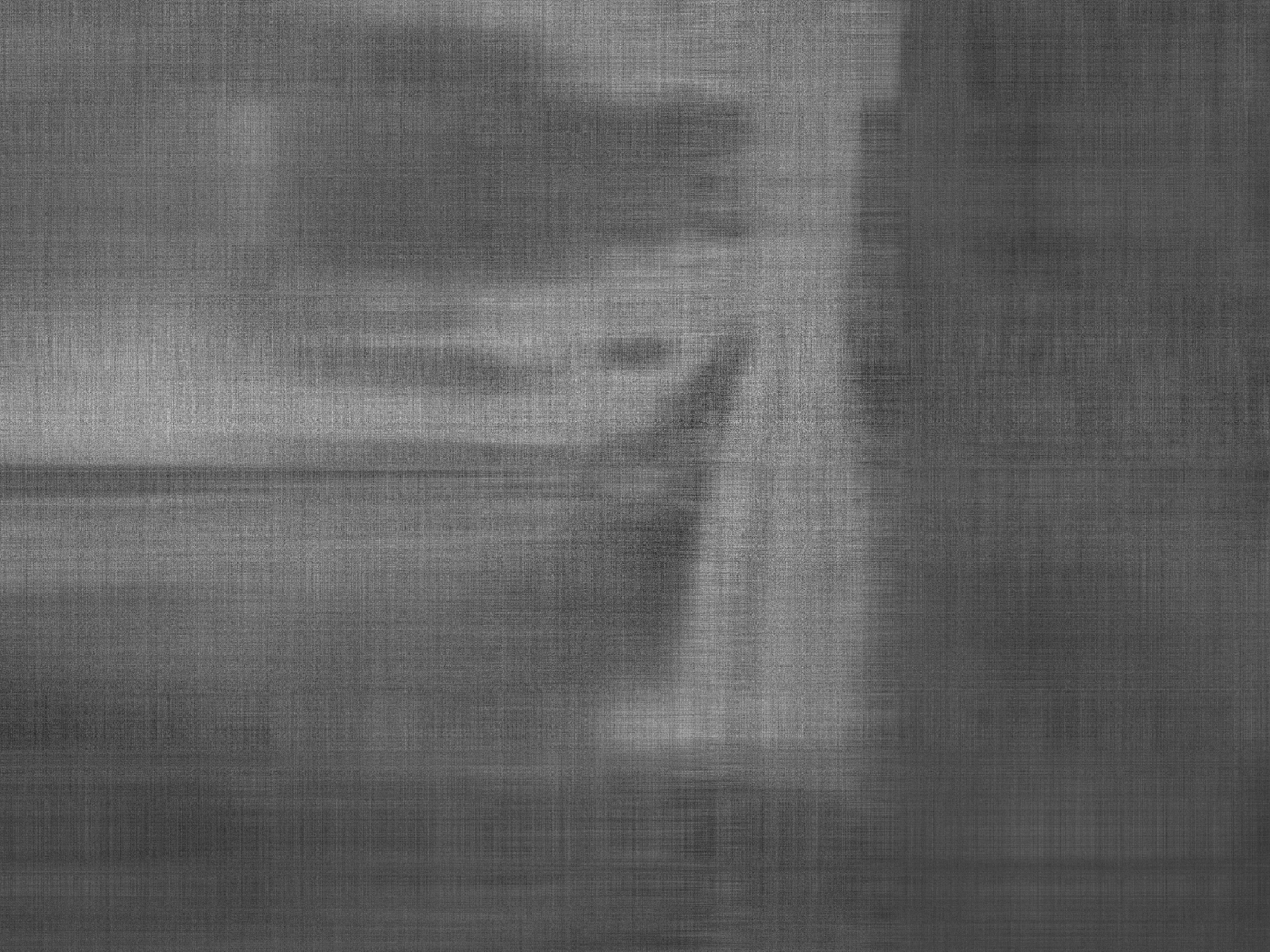}
\caption{Na\"ive Method}
\end{subfigure}
\begin{subfigure}{0.3\textwidth}
\includegraphics[width=\textwidth]{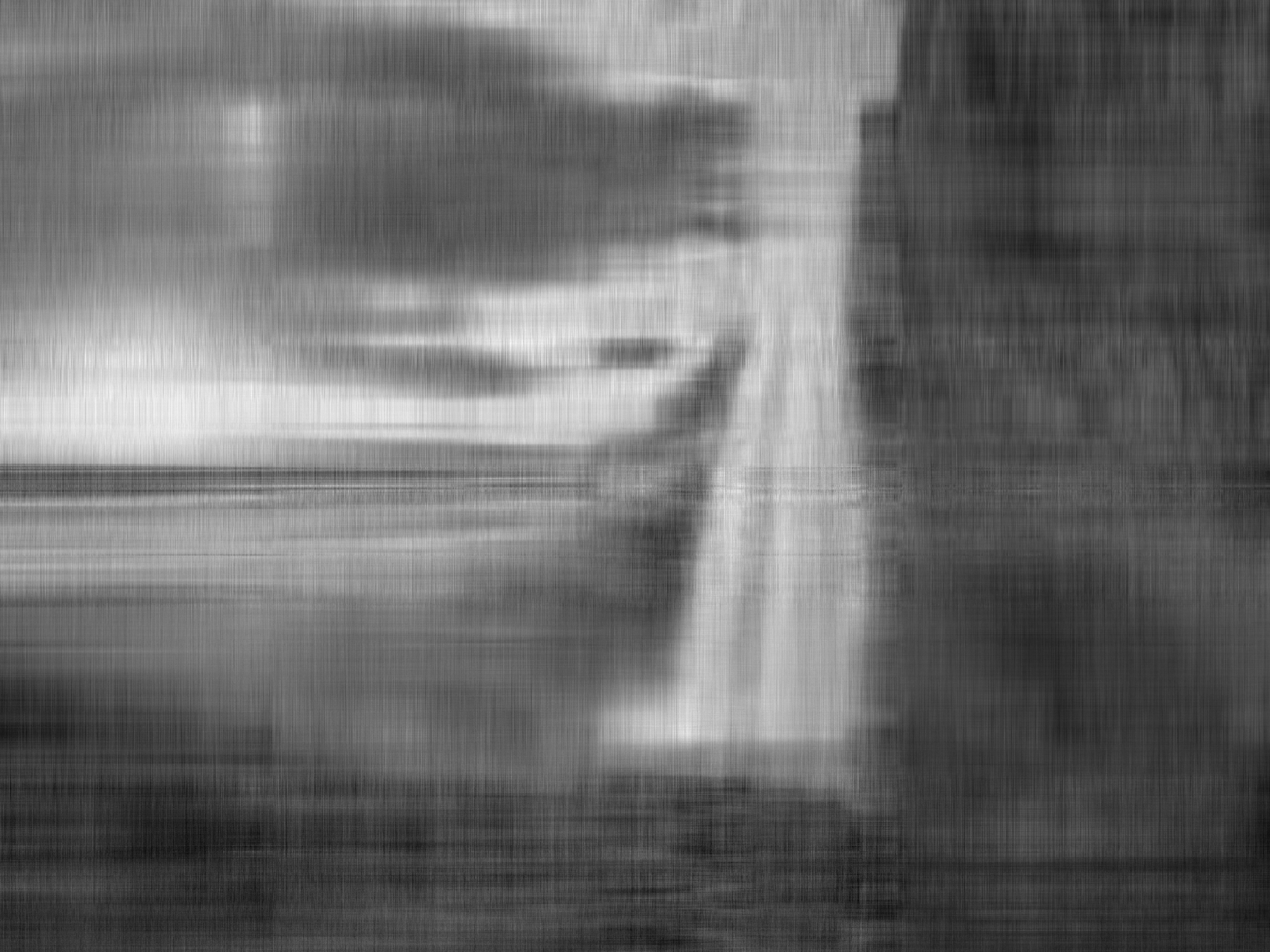}
\caption{Improved Method}   
\end{subfigure}
\caption{The exact best rank-30 approximation of the original greyscale image is given in panel (a). Panels (b) and (c) are the output by the na\"ive method and the improved method.}
\label{fig:waterfall}
\end{figure}

%
%

\section{Proofs}
\begin{proof}
[Proof of Lemma \ref{le:lowrankappr}]

By definition,
$$
\|B_r-B\|_{\rm F}^2\le \inf_{\rank(M)\le r}\|M-B\|_{\rm F}^2.
$$
Observe that
$$
\|M-B\|_{\rm F}^2=\|M-A\|_{\rm F}^2+\|E\|_{\rm F}^2-2\langle E, M-A\rangle.
$$
We get
$$
\|B_r-A\|_{\rm F}^2\le \|M-A\|_{\rm F}^2+2\langle E, B_r-M\rangle\le \|M-A\|_{\rm F}^2+2\|E\|\|B_r-M\|_\ast,
$$
where we replaced $A$ with $B-E$ in the first inequality. Recall that
$$
\rank(B_r-M)\le \rank(B_r)+\rank(M)\le 2r.
$$
By Cauchy-Schwartz inequality,
$$
\|B_r-M\|_\ast\le \sqrt{2r}\|B_r-M\|_{\rm F}.
$$
Therefore,
$$
\|B_r-A\|_{\rm F}^2\le \|M-A\|_{\rm F}^2+2\sqrt{2r}\|E\|\|B_r-M\|_{\rm F}.
$$
In the case when $\rank(A)\le r$, we can take $M=A$ to yield
$$
\|B_r-A\|_{\rm F}^2\le 8r\|E\|^2.
$$
Now consider the case when $\rank(A)>r$. Taking $M=A_r$ yields
$$
\|B_r-A\|_{\rm F}^2\le \|A_r-A\|_{\rm F}^2+2\sqrt{2r}\|E\|\|B_r-A_r\|_{\rm F}.
$$
Observe that
$$
\|B_r-A\|_{\rm F}^2=\|B_r-A_r\|_{\rm F}^2+\|A_r-A\|_{\rm F}^2+2\langle B_r, A_r-A\rangle.
$$
We have
$$
\|B_r-A_r\|_{\rm F}^2\le 2\sqrt{2r}\|E\|\|B_r-A_r\|_{\rm F}+2\langle B_r, A-A_r\rangle.
$$
Let $A_r=U_r\Sigma_r V_r^\top$ and $B_r=\tilde{U}_r\tilde{\Sigma}_r \tilde{V}_r^\top$ be their respective singular value decomposition. Then
\begin{eqnarray*}
\langle B_r, A-A_r\rangle&=&\langle P_{U_r}^\perp B_rP_{V_r}^\perp, A-A_r\rangle\\
&\le& \sigma_{r+1}(A)\|P_{U_r}^\perp B_rP_{V_r}^\perp\|_\ast\\
&\le& r\sigma_{r+1}(A)\|P_{U_r}^\perp P_{\tilde{U}_r} B_rP_{\tilde{V}_r}P_{V_r}^\perp\|\\
&\le& r\sigma_{r+1}(A)\sigma_1(B)\|P_{U_r}^\perp P_{\tilde{U}_r}\|\|P_{\tilde{V}_r}P_{V_r}^\perp\|.
\end{eqnarray*}
By Davis-Kahan-Wedin's Theorems,
$$
\max\{\|P_{U_r}^\perp P_{\tilde{U}_r}\|,\|P_{\tilde{V}_r}P_{V_r}^\perp\|\}\le {4\|E\|\over \sigma_r(A)-\sigma_{r+1}(A)},
$$
$$|\sigma_1(A)-\sigma_1(B)| \le \|E\| < \sigma_1(A), $$
so that
$$
\langle B_r, A-A_r\rangle\le {16r\sigma_{r+1}(A)\sigma_1(B)\|E\|^2\over [\sigma_r(A)-\sigma_{r+1}(A)]^2}\le {32r\sigma_{r+1}(A)\sigma_1(A)\|E\|^2\over [\sigma_r(A)-\sigma_{r+1}(A)]^2}.
$$
This implies that
$$
\|B_r-A_r\|_{\rm F}\le 8\sqrt{r}\|E\| \left({1\over \sqrt{2}}+{\sqrt{\sigma_{r+1}(A)\sigma_1(A)}\over \sigma_r(A)-\sigma_{r+1}(A)}\right).
$$
\end{proof}

\begin{proof}[Proof of Lemma \ref{le:remainder}] 
The proof relies on the following concentration bound.

\begin{lemma}
\label{le:concentration}
Let $B\in \RR^{d_1\times d_2}$ be a fixed matrix and each entry of $\Omega$ is independently sampled from binomial trails with probability given by \eqref{eq:defp}. Then there exists a numerical constant $C_0>0$ such that
$$
\left\|\calP_\Omega(B)-B\right\|\le 4\|A\|_{\rm F}\sqrt{d\over n}\max_{\substack{1\le i\le d_1\\1\le j \le d_2}}\left\{{\| B_{i\cdot}\|\over \|A_{i\cdot}\|},{\|B_{\cdot j}\|\over \|A_{\cdot j}\|}\right\}+ C_0\max_{\substack{i,j:\\ 0<p_{ij}<1}}\left\{ |b_{ij}|\over p_{ij}\right\}\sqrt{\log (d)+t},
$$
with probability at least $1-e^{-t}$, where $d=\max\{d_1,d_2\}$.
\end{lemma}

In particular, consider applying Lemma \ref{le:concentration} to $B=A-A_r$. It is not hard to see that
$$
{\|B_{i\cdot}\|^2\over \|A_{i\cdot}\|^2}=1-{\|(A_r)_{i\cdot}\|^2\over \|A_{i\cdot}\|^2} \le \nu_r(A)^2,
$$
and
$$
{\|B_{\cdot j}\|^2\over \|A_{\cdot j}\|^2}=1-{\|(A_r)_{\cdot j}\|^2\over \|A_{\cdot j}\|^2} \le \nu_r(A)^2.
$$
By definition of $\nu_r(A)_\infty$, we get
\begin{eqnarray*}
&&\max_{(i,j): 1>  p_{ij}>0}\left\{|b_{ij}|\over p_{ij}\right\}\\
&\le& \nu_r(A)_\infty\cdot\max_{(i,j): 1>  p_{ij}>0}\left\{|a_{ij}| + (\|A_{i \cdot}\|^2+\|A_{\cdot j}\|^2)/\|A\|_{\rm F}\over p_{ij}\right\}\\
&\leq &\nu_r(A)_\infty \cdot \max\left\{{3\|A\|_{\ell_1}\over n},{3d\|A\|_{\rm F}\over n} \right\}\\
&\leq &\nu_r(A)_\infty \cdot {3d\|A\|_{\rm F}\over n} ,
\end{eqnarray*}
where we used the fact that $\|A\|_{\ell_1} \le d\|A\|_{\rm F}$ and thus the first claim follows.

For the second claim, let $A = U\Sigma V^\top$ be its singular value decomposition,
\begin{eqnarray*}
| b_{ij}|&=&\left|\sum_{k>r}\sigma_k(A)u_{ik}v_{jk}\right|\\
&\le& |a_{ij}|+\left|\sum_{k\le r}\sigma_k(A)u_{ik}v_{jk}\right|\\
&\le& |a_{ij}|+\left(\sum_{k\le r}\sigma_k(A)u_{ik}^2\right)^{1/2}\left(\sum_{k\le r}\sigma_k(A)v_{jk}^2\right)^{1/2}\\
&\le& |a_{ij}|+{1\over \sigma_r(A)}\left(\sum_{k\le r}\sigma_k(A)^2u_{ik}^2\right)^{1/2}\left(\sum_{k\le r}\sigma_k(A)^2v_{jk}^2\right)^{1/2}\\
&=& |a_{ij}|+ {1\over \sigma_r(A)}\|(A_r)_{ i\cdot}\|\|(A_r)_{ \cdot j}\|.
\end{eqnarray*}
Hence,
$$
\nu_r(A)_\infty \le  1 + {\|A\|_{\rm F}\over \sigma_r(A)}\max\left\{{\|(A_r)_{i\cdot}\|^2\over \|A_{i\cdot}\|^2},{\|(A_r)_{\cdot j}\|^2\over \|A_{\cdot j}\|^2}\right\}.
$$
This implies that
\begin{eqnarray*}
&&\left\|\calP_\Omega(A-A_r)-(A-A_r)\right\| \\
&\le&C_0\left(1+{\mu_r(A)^2\|A\|_{\rm F}\over \sigma_r(A)}\right){d\|A\|_{\rm F}\sqrt{\log d+t}\over n}\\
 &&+4\nu_r(A)\|A\|_{\rm F}\sqrt{d\over n},
\end{eqnarray*}
with probability at least $1-e^{-t}$.
\end{proof}

\begin{proof}[Proof of Theorem \ref{th:main}]

In the rest of the proof, we shall omit $A$ in $\sigma_r(A)$ for brevity. Let
$$R_t = \argmin_{RR^\top = R^\top R = I} \left\|F_t - F^*R \right\|_\textup{F}.$$
We shall then write $S_t = X^*R_t$, and $T_t = Y^*R_t$. Moreover, denote by
$$\Delta_{X_t} =X_t- S_t,\qquad \Delta_{Y_t} =Y_t- T_t,\qquad \Delta_{F_t} =F_t- F^*R_t.$$
Our analysis relies on the follow two technical lemmas. 

\begin{lemma}\label{b1}
Assume that $\beta \leq \sqrt{\sigma_r/(1-\nu_r(A)^2)}$ and $(X_t,Y_t)\in \calF_\beta$ such that
$$\|\Delta_{F_t}\|_{\rm F} \leq \frac{\sigma_r-\sigma_{r+1}}{20\sqrt{\sigma_1}}.$$
There exists a numerical constant $C_1>0$ such that, with probability at least $1-d^{-\alpha}$, if
\begin{equation}
\label{eq:b1cond}
\max\left\{\|X_tY_t^\top-A_r\|_{\rm F},\sqrt{\sigma_r}\| \Delta_{F_t}\|_{\rm F}\right\}>\frac{31\sqrt{r}}{10(1-\sigma_{r+1}/\sigma_r)}\left\|\calP_\Omega(N_r)-N_r\right\|,
\end{equation}
then
\begin{align*}
\langle \nabla_{X} f(X_t,Y_t),\Delta_{X_t}\rangle + \langle \nabla_Y f(X_t,Y_t),\Delta_{Y_t}\rangle 
\geq {19(\sigma_r-\sigma_{r+1})\over 100}\|\Delta_{F_t}\|_{\rm F}^2;
\end{align*}
and otherwise,
\begin{align*}
\langle \nabla_{X} f(X_t,Y_t),\Delta_{X_t}\rangle + \langle \nabla_Y f(X_t,Y_t),\Delta_{Y_t}\rangle
\geq  \frac{69(\sigma_r-\sigma_{r+1})}{100}\|\Delta_{F_t}\|_{\rm F}^2 - {5r\| \calP_\Omega(N_r)-N_r\|^2\over (1-\sigma_{r+1}/\sigma_r)},
\end{align*}
provide that $n \geq C_1(1+\alpha)d(\log d)\sigma_r^2/\left[\beta^4(1-\nu_r(A)^2)^2\right] \mu_r^2(A)\|A_r\|_{\rm F}^2 \|A\|_{\rm F}^2/(\sigma_r-\sigma_{r+1})^4$.
\end{lemma}

Lemma \ref{b1} verifies the so-called local descent condition from \cite{chen2015fast}. It states that, if $\|\Delta_{F_t}\|_\textup{F} \geq O(\| \calP_\Omega(N_r)-N_r\|)$, $\nabla f(X_t,Y_t)$ will have an acute angle with $\Delta_{F_t}$ and $f(X_t,Y_t)$ shows similar behavior to a convex function. It implies that if $\nabla f \approx 0$, then it is necessarily true that
$$\|\Delta_{F_t}\|_{\rm F}=O\left(\left\|\calP_\Omega(N_r)-N_r\right\|\right),$$
so that, to bound $\|\Delta_{F_t}\|_\textup{F}$, it suffices to do so for $\| \calP_\Omega(N_r)-N_r\|$.

\begin{lemma}\label{b2}
Under the conditions in Lemma \ref{b1}, there exists a numerical constant $C_1>0$ such that
\begin{align*}
\| \nabla_X f(X_t,Y_t)\|_{\rm F}^2 + \| \nabla_Y f(X_t,Y_t)\|_{\rm F}^2\leq \frac{81\|A_r\|_{\rm F}^4}{\beta^4(1-\nu_r(A)^2)^2}\|\Delta_{F_t}\|_\textup{F}^2 +18r\sigma_1\| \calP_\Omega(N_r)-N_r\|^2,
\end{align*}
with probability at least $1-d^{-\alpha}$ provide $n \geq C_1\alpha d\log d\|A\|_{\rm F}^2/\|A_r\|_{\rm F}^2$.
\end{lemma}

We shall now use these lemmas to prove that for any $t$
\begin{equation}
\label{inducta}
\|\Delta_{F_t}\|_{\rm F} \leq \frac{\sigma_r-\sigma_{r+1}}{20\sqrt{\sigma_1}},
\end{equation}
and
\begin{equation}
\label{inductc}
\delta_t 
\leq
\left[ 1-\frac{1}{5}\eta\left(\sigma_r-\sigma_{r+1}\right) \right]^{1/2}\cdot \delta_{t-1}+\frac{4\sqrt{r /\sigma_r}}{1-\sigma_{r+1}/\sigma_r}\left\|\calP_\Omega(N_r)-N_r\right\|.
\end{equation}



We shall do so by induction. In fact, when $t=0$, it suffices to verify \eqref{inducta}. In light of Lemma \ref{le:remainder} and Lemma \ref{le:lowrankappr}, we have
$$\left\| X_0Y_0^\top -A_r \right\|_\textup{F} \leq O\left(\frac{\sigma_r-\sigma_{r+1}}{\sqrt{\sigma_1/\sigma_r}}\right),$$
with probability at least $1-d^{-\alpha}$ provided that 
$$n\ge C_1(1+\alpha)rd(\log d)\cdot{(\sigma_1/\sigma_r)\|A\|_{\rm F}^2\over (\sigma_r-\sigma_{r+1})^2}\cdot\left(1+{\sigma_{r+1}\sigma_1\over (\sigma_r-\sigma_{r+1})^2}\right).$$
Then, \eqref{inducta} follows immediately from Lemma 6 and Lemma 42 of \cite{ge2017no}.

Now assume that \eqref{inducta} and \eqref{inductc} hold for $t$, we show that the same is true for $t+1$. We shall first verify \eqref{inductc}. Because the projection $\calP_{\calF_\beta}$ is contractive, we have 
\begin{eqnarray*}
\|F_{t+1}-F^*R_t\|_\textup{F}^2 & \le& \|X_t-\eta\nabla_X f(X_t,Y_t)-S_t\|_\textup{F}^2 +  \|Y_t-\eta\nabla_Y f(X_t,Y_t)-T_t\|_\textup{F}^2\\
&=& \|\Delta_{F_t}\|_\textup{F}^2 +\eta^2(\| \nabla_X f\|_{\rm F}^2 + \| \nabla_Y f\|_{\rm F}^2) - 2\eta \langle \nabla_{X} f(X_t,Y_t),\Delta_{X_t}\rangle \\
&&\hskip50pt -2\eta \langle \nabla_Y f(X_t,Y_t),\Delta_{Y_t}\rangle \\
&\leq& \|\Delta_{F_t}\|_\textup{F}^2 +{  81\eta^2\|A_r\|_{\rm F}^4 \over \beta^4(1-\nu_r(A)^2)^2}\|\Delta_{F_t}\|_\textup{F}^2 +18\eta^2r\sigma_1\| \calP_\Omega(N_r)-N_r\|^2\\
&&\hskip50pt - 2\eta \langle \nabla_{X} f(X_t,Y_t),\Delta_{X_t}\rangle -2\eta \langle \nabla_Y f(X_t,Y_t),\Delta_{Y_t}\rangle ,\label{b.4}
\end{eqnarray*}
where we used Lemma \ref{b2}.

By Lemma \ref{b1}, if \eqref{eq:b1cond} holds, then
\begin{eqnarray}
\delta_{t+1}^2 &\leq &\|F_{t+1}-F^*R_t\|_\textup{F}^2 \nonumber\\
&\leq& \|\Delta_{F_t}\|_\textup{F}^2 +{  81\eta^2\|A_r\|_{\rm F}^4 \over \beta^4(1-\nu_r(A)^2)^2}\|\Delta_{F_t}\|_\textup{F}^2 \nonumber\\
&&+18\eta^2r\sigma_1\| \calP_\Omega(N_r)-N_r\|^2 - \eta {19(\sigma_r-\sigma_{r+1})\over 50}\|\Delta_{F_t}\|_{\rm F}^2  \nonumber\\
&\leq&  \left[ 1-\frac{1}{5}\eta(\sigma_r -\sigma_{r+1}) \right]\delta_t^2,\label{eq:main1}
\end{eqnarray}
where we used the condition that $\eta \leq (\sigma_r-\sigma_{r+1})\beta^4(1-\nu_r(A)^2)^2/(500\|A_r\|_{\rm F}^4)$.

On the other hand, if \eqref{eq:b1cond} does not hold, then by Lemma \ref{b1},
\begin{eqnarray*}
\delta_{t+1}^2 &\leq &\|F_{t+1}-F^*R_t\|_\textup{F}^2 \\
&\leq&  \frac{31^2 \cdot r/\sigma_r \cdot \|\calP_\Omega(N_r)-N_r\|^2}{10^2(1-\sigma_{r+1}/\sigma_r)^2} + \eta{5r\| \calP_\Omega(N_r)-N_r\|^2\over (1-\sigma_{r+1}/\sigma_r)}\\
&&+{  81\eta^2\|A_r\|_{\rm F}^4 \over \beta^4(1-\nu_r(A)^2)^2}\|\Delta_{F_t}\|_\textup{F}^2 +18\eta^2r\sigma_1\| \calP_\Omega(N_r)-N_r\|^2\\
&&- \eta {69(\sigma_r-\sigma_{r+1})\over 50}\|\Delta_{F_t}\|_{\rm F}^2\\
&\leq&  \frac{4^2 \cdot r/\sigma_r \cdot \|\calP_\Omega(N_r)-N_r\|^2}{(1-\sigma_{r+1}/\sigma_r)^2},
\end{eqnarray*}
where we again used $\eta \leq (\sigma_r-\sigma_{r+1})\beta^4(1-\nu_r(A)^2)^2/(500\|A_r\|_{\rm F}^4)$.

In light of above two inequalities, we have,
$$\delta_{t+1} \le \max\left\{\delta_{t}, \frac{4\sqrt{r/\sigma_r}}{1-\sigma_{r+1}/\sigma_r}\cdot\|\calP_\Omega(N_r)-N_r\|\right\}.$$
It is clear that \eqref{inducta} also continues to hold for $t+1$ in light of the inequality above and the Lemma \ref{le:remainder}. The first claim of Theorem \ref{th:main} follows immediately.

The second claim also follows, in light of the following bound:
\begin{lemma}
\label{le:apprest}
Under the assumptions of Theorem \ref{th:main},
$$
\|X_tY_t^\top-A_r\|_{\rm F}^2\le 3\sigma_1(A)\delta_t^2.
$$
\end{lemma}
If the second claim doesn't hold, then \eqref{eq:b1cond} holds and we have \eqref{eq:main1}. Combined with Lemma \ref{le:apprest}, there is a contradiction.
\end{proof}

\begin{proof}[Proof of Lemma \ref{b1}]
For brevity, we shall omit the subscript $t$ in what follows. Denote by $\kappa_r = \sigma_{r+1}/\sigma_r$. Observe that
$$
\nabla_X f = \calP_{\Omega}\left(XY^\top - A\right)Y +{1\over 2}X(X^\top X - Y^\top Y),
$$
and
$$
\nabla_Y  f= \left[\calP_{\Omega}\left(XY^\top - A\right)\right]^\top X -{1\over 2}Y(X^\top X - Y^\top Y) .
$$
Therefore,
\begin{eqnarray}
&&\langle \nabla_X f,\Delta_X\rangle + \langle \nabla_Y f,\Delta_Y\rangle \nonumber\\
&=& \left\langle \calP_\Omega\left(XY^\top-A\right)Y,\Delta_X\right\rangle + \left\langle \left[\calP_\Omega\left(X Y^\top-A\right)\right]^\top X,\Delta_Y\right\rangle\nonumber\\
&& \hskip 75pt+ {1\over 2} \left\langle X^\top X - Y^\top Y ,X^\top(\Delta_X) - (\Delta_Y)^\top Y\right\rangle\nonumber\\
 &=& \left\langle \calP_\Omega \left(XY^\top-ST^\top\right),(\Delta_X)Y^\top+X(\Delta_Y)^\top\right\rangle \label{c.1} \\
&&\hskip 75pt + {1\over 2} \left\langle X^\top X - Y^\top Y ,X^\top(\Delta_X) - (\Delta_Y)^\top Y\right\rangle\label{c.3}\\
&& \hskip 75pt- \left\langle \calP_\Omega(N_r),(\Delta_X)Y^\top+X(\Delta_Y)^\top\right\rangle \label{c.2}.
\end{eqnarray}
We now bound each of the term on the rightmost hand side to show that it can be lower bounded by $\|\Delta_{F_t}\|_{\rm F}^2$.

\paragraph{\indent Bounding \eqref{c.1}:} The term \eqref{c.1} can be bounded in a similar way as \cite{ge2016matrix}. However, since we are dealing with different $p_{ij}$ instead of uniform sampling in \cite{ge2016matrix}, new concentration inequalities are needed. 

Denote by $D_1=S\Delta_Y^\top+\Delta_XT^\top$, $D_2=\Delta_X \Delta_Y^\top$, and
$$D_3 = \Delta_X^\top S + S^\top\Delta_X - \Delta_Y^\top T-T^\top \Delta_Y.$$
It is clear that
$$XY^\top-ST^\top=D_1+D_2,$$
and
$$(\Delta_X)Y^\top+X(\Delta_Y)^\top = D_1+2D_2.$$
Thus, by Cauchy-Schwartz inequality
\begin{eqnarray*}
 &&\langle \calP_\Omega(XY^\top-ST^\top),(\Delta_X)Y^\top+X(\Delta_Y)^\top\rangle  \\
& =&  \langle \calP_\Omega(D_1+D_2),D_1+2D_2\rangle \\
& =& \langle \calP_\Omega(D_1),D_1\rangle + 2 \langle \calP_\Omega(D_2),D_2\rangle- 3\langle \calP_\Omega(D_1), D_2 \rangle  \\
& \geq& \langle \calP_\Omega(D_1),D_1\rangle - 3\langle \calP_\Omega(D_1), D_2 \rangle \\
& \geq&\langle \calP_\Omega(D_1),D_1\rangle -3\sqrt{\langle \calP_\Omega(D_1),D_1\rangle\langle \calP_\Omega(D_2),D_2\rangle},
\end{eqnarray*}
We shall make use of the following concentration inequalities:
\begin{lemma}
\label{c1}
Under the assumptions of Theorem \ref{th:main}, with probability at least $1-d^{-\alpha}$,
$$\left| \langle \calP_\Omega(D_1),D_1\rangle - \|D_1\|^2_\textup{F} \right| \leq \frac{1-\kappa_{r}}{20}\cdot \|D_1\|_\textup{F}^2,$$
and
$$\langle \calP_\Omega(D_2),D_2\rangle \leq \left(\frac{1-\kappa_{r}}{20}\right)^2\cdot \sigma_r\|\Delta_F\|_\textup{F}^2.$$
\end{lemma}
In light of Lemma \ref{c1}, we have
\begin{align}
\label{c1bound}
& \langle \calP_\Omega(XY^\top-ST^\top),(\Delta_X)Y^\top+X(\Delta_Y)^\top\rangle \nonumber \\
& \geq \left(1- \frac{1-\kappa_r}{20}\right)\|D_1\|_\textup{F}^2-3 \sqrt{\sigma_r}\|\Delta_F\|_\textup{F} \cdot \frac{1-\kappa_r}{20}\cdot \sqrt{1+\frac{1-\kappa_r}{20}}\|D_1\|_\textup{F}.
\end{align}

\paragraph{\indent Bounding \eqref{c.3}:} It is not hard to see that $S^\top S = T^\top T$. Thus, by rewriting $X = S + \Delta_X$ and $Y = T + \Delta_Y$, we have
\begin{align*}
&{1\over 2} \left\langle X^\top X - Y^\top Y ,X^\top(\Delta_X) - (\Delta_Y)^\top Y\right\rangle \\
=& {1\over 2} \left\langle X^\top X - Y^\top Y,\Delta_X^\top\Delta_X -\Delta_Y^\top \Delta_Y\right\rangle +{1\over 2} \left\langle\Delta_X^\top\Delta_X -\Delta_Y^\top \Delta_Y,  S^\top \Delta_X - \Delta_Y^\top T \right\rangle \\
&\hskip 50pt+ {1\over 2} \left\langle \Delta_X^\top S + S^\top\Delta_X - \Delta_Y^\top T-T^\top \Delta_Y,S^\top \Delta_X - \Delta_Y^\top T \right\rangle\\
=& {1\over 2} \left\langle \Delta_X^\top S + 2S^\top\Delta_X - 2\Delta_Y^\top T-T^\top \Delta_Y+ \Delta_X^\top\Delta_X -\Delta_Y^\top \Delta_Y,\Delta_X^\top\Delta_X -\Delta_Y^\top \Delta_Y\right\rangle \\
&\hskip 50pt+ {1\over 4} \left\| \Delta_X^\top S + S^\top\Delta_X - \Delta_Y^\top T-T^\top \Delta_Y\right\|_{\rm F}^2,
\end{align*}
where we used the result that $\Delta_X^\top S + \Delta_Y^\top T$ is symmetric from Lemma 6 of \cite{ge2017no}. The first term on the righthand side can be bounded by
\begin{eqnarray*}
&&\left| {1\over 2}\left\langle \Delta_X^\top S + 2S^\top\Delta_X - 2\Delta_Y^\top T-T^\top \Delta_Y+ \Delta_X^\top\Delta_X -\Delta_Y^\top \Delta_Y,\Delta_X^\top\Delta_X -\Delta_Y^\top \Delta_Y\right\rangle \right| \\
&\leq&{1\over 2}(3\|S\|\|\Delta_X\|_{\rm F}+3\|T\|\|\Delta_Y\|_{\rm F}+ \|\Delta_X^\top\Delta_X -\Delta_Y^\top \Delta_Y\|_{\rm F})\|\Delta_X^\top\Delta_X -\Delta_Y^\top \Delta_Y\|_{\rm F} \\
&\leq&2\sqrt{\sigma_1}(\|\Delta_X\|_{\rm F}+\|\Delta_Y\|_{\rm F})(\|\Delta_X\|_{\rm F}^2+\|\Delta_Y\|_{\rm F}^2)\\
& \leq& 3\sqrt{\sigma_1}\|\Delta_F\|_{\rm F}^3,
\end{eqnarray*}
where the last inequality follow from the fact that $\max\{ \|\Delta_X\|, \|\Delta_Y\| \} \leq \|\Delta_F\|_{\rm F} \le \sqrt{\sigma_1}$.
\paragraph{\indent Bounding \eqref{c.2}:}
Observe that
\begin{align}
\left\langle \calP_\Omega(N_r),(\Delta_X)Y^\top+X(\Delta_Y)^\top \right\rangle  
= &\left\langle \calP_\Omega(N_r)-N_r,(\Delta_X)Y^\top+X(\Delta_Y)^\top \right\rangle  \nonumber\\
&+ \left\langle N_r,(\Delta_X)Y^\top+X(\Delta_Y)^\top \right\rangle.\label{c.4}
\end{align}
The first term on the righthand side can be bounded by
\begin{eqnarray*}
&& \left\langle \calP_\Omega(N_r)-N_r,(\Delta_X)Y^\top+X(\Delta_Y)^\top \right\rangle  \nonumber\\
&\leq &  \left\| \calP_\Omega(N_r)-N_r \right\| \left\| (\Delta_X)Y^\top+X(\Delta_Y)^\top \right\|_*  \nonumber\\
&= &  \left\| \calP_\Omega(N_r)-N_r \right\| \sqrt{2r}\left\| (\Delta_X)Y^\top+X(\Delta_Y)^\top \right\|_\textup{F}  \nonumber\\
&\leq &  \sqrt{2r}\left\| \calP_\Omega(N_r)-N_r \right\| \left(\left\|D_1 \right\|_\textup{F} + 2\|D_2\|_{\rm F}\right)  \nonumber\\
&\leq & \sqrt{2r} \left\| \calP_\Omega(N_r)-N_r \right\| (\|D_1\|_\textup{F} + \|\Delta_F\|_{\rm F}^2) ,
\end{eqnarray*}
where we made use of the facts that, for two arbitrary matrices $A$ and $B$,
$$\langle A,B \rangle \leq \|A\| \|B\|_*,\qquad \|A\|_*^2 \leq \rank(A)\|A\|_\textup{F}^2.$$
Note that by definition $S^\top N_r = 0$ and $N_r T =0$. This implies that
\begin{eqnarray*}
&&\left\langle N_r,(\Delta_X)Y^\top+X(\Delta_Y)^\top \right\rangle  \nonumber\\
&=&  \tr( (\Delta_X)^\top N_rY) +  \tr( X^\top N_r (\Delta_Y)) \nonumber \\
&=&  2 \cdot \tr(\Delta_X^\top N_r \Delta_Y)  \\
&\leq& 2\sigma_{r+1}\|\Delta_X\|_\textup{F}\|\Delta_Y\|_\textup{F}\nonumber \\ 
&\leq  &\sigma_{r+1}\|\Delta_F\|_\textup{F}^2 .
\end{eqnarray*}
In summary, we have
\begin{eqnarray*}
&&\left\langle \calP_\Omega(N_r)-N_r,(\Delta_X)Y^\top+X(\Delta_Y)^\top \right\rangle \\
&\leq&    \sqrt{2r} \left\| \calP_\Omega(N_r)-N_r \right\| (\|D_1\|_\textup{F} + \|\Delta_F\|_{\rm F}^2) + \sigma_{r+1}\|\Delta_F\|_\textup{F}^2.
\end{eqnarray*}
Together, the bounds for \eqref{c.1}, \eqref{c.3} and \eqref{c.2} imply that
\begin{eqnarray*}
&& \langle \nabla_X f,\Delta_X\rangle + \langle \nabla_Y f,\Delta_Y\rangle \\
 &\geq & \|D_1\|_\textup{F}^2 +  {1\over 4}\|D_3\|_{\rm F}^2- \sigma_{r+1}\|\Delta_F\|_\textup{F}^2 -\sqrt{2r} \left\| \calP_\Omega(N_r)-N_r \right\|(\|D_1\|_\textup{F} + \|\Delta_F\|_{\rm F}^2)\\
&&  - \frac{1-\kappa_r}{20}\|D_1\|_\textup{F}^2-3\sqrt{\sigma_r} \|\Delta_F\|_\textup{F} \cdot \frac{1-\kappa_r}{20}\cdot \sqrt{1+\frac{1-\kappa_r}{20}}\|D_1\|_\textup{F} -3\sqrt{\sigma_1}\|\Delta_F\|_{\rm F}^3 .
\end{eqnarray*}
Denote by 
\begin{eqnarray*}
\epsilon^2 &= &\|\Delta_XT^\top +S\Delta_Y^\top\|_{\rm F}^2 + {1\over 4}\| \Delta_X^\top S + S^\top\Delta_X - \Delta_Y^\top T-T^\top \Delta_Y\|_{\rm F}^2\\
&= &\|\Delta_X T^\top\|_{\rm F}^2 +\|S \Delta_Y^\top\|_{\rm F}^2 + \langle \Delta_X T^\top, S\Delta_Y^\top \rangle + {1\over 2}\| \Delta_X^\top S\|_{\rm F}^2 + {1\over 2}\| \Delta_Y^\top T\|_{\rm F}^2\\
&&\hskip50pt  + {1\over 2}\langle\Delta_X^\top S, S^\top\Delta_X \rangle + {1\over 2}\langle\Delta_Y^\top T, T^\top \Delta_Y \rangle - \langle \Delta_X^\top S ,\Delta_Y^\top T \rangle \\
&= & \|\Delta_X T^\top\|_{\rm F}^2 +\|S \Delta_Y^\top\|_{\rm F}^2  + {1\over 2}\| \Delta_X^\top S-\Delta_Y^\top T \|_{\rm F}^2 \\
&&\hskip50pt + {1\over 2}\langle \Delta_X^\top S+\Delta_Y^\top T ,  S^\top \Delta_X+T^\top \Delta_Y \rangle \\
&= & \|\Delta_X T^\top\|_{\rm F}^2 +\|S \Delta_Y^\top\|_{\rm F}^2  + {1\over 2}\| \Delta_X^\top S-\Delta_Y^\top T \|_{\rm F}^2+ {1\over 2}\| \Delta_X^\top S+\Delta_Y^\top T \|_{\rm F}^2,
\end{eqnarray*}
where we again used the result that $\Delta_X^\top S + \Delta_Y^\top T$ is symmetric from Lemma 6 of \cite{ge2017no} and the fact that $ \langle \Delta_X T^\top, S\Delta_Y^\top \rangle = \langle\Delta_X^\top S, T^\top\Delta_Y \rangle$.

This implies
\begin{align}
\label{epsilon:low}
\epsilon^2 \geq \sigma_r\|\Delta_F\|_{\rm F}^2.
\end{align}
Note that
\begin{align}
\label{deltaf}
\|\Delta_F\|_{\rm F}^2= \|\Delta_F\|_{\rm F}\cdot\|\Delta_F\|_{\rm F} \leq \frac{1-\kappa_{r}}{20\sqrt{\sigma_1}}\sigma_r\|\Delta_F\|_{\rm F}
\leq \frac{1-\kappa_{r}}{20}\sqrt{\sigma_r}\|\Delta_F\|_{\rm F} \leq \frac{1-\kappa_{r}}{20}\epsilon.
\end{align}
We get
\begin{eqnarray}
\label{gradient:lowbound}
&& \langle \nabla_X f,\Delta_X\rangle + \langle \nabla_Y f,\Delta_Y\rangle \nonumber\\
 &\geq & \epsilon^2 (1  -\kappa_r)-\sqrt{2} \left\| \calP_\Omega(N_r)-N_r \right\|(\|D_1\|_\textup{F} + \|\Delta_F\|_{\rm F}^2)- {31\over 5}\cdot\frac{1-\kappa_r}{20}\epsilon^2 \\
&\geq & \frac{69(1-\kappa_{r})}{100}\epsilon^2- 1.5\sqrt{r} \left\| \calP_\Omega(N_r)-N_r \right\|\epsilon. \nonumber
\end{eqnarray}
In light of \eqref{deltaf}, if
\begin{equation}
\label{condition}
\max\left\{\|D_1+D_2\|_{\rm F},\sqrt{\sigma_r}\| \Delta_{F}\|_{\rm F}\right\}>\frac{31\sqrt{r}}{10(1-\kappa_{r})}\left\|\calP_\Omega(N_r)-N_r\right\|,
\end{equation}
then
\begin{eqnarray*}
\|D_1+D_2\|_{\rm F}^2 &\leq& (\|D_1\|_{\rm F} + 0.5\|\Delta_F\|_{\rm F}^2)^2 \\
&\leq& (\|D_1\|_{\rm F} +\epsilon/40)^2\\ 
&\leq& \|D_1\|_{\rm F}^2 + \|D_1\|_{\rm F}\epsilon/20 + \epsilon^2/400\\
&<& 1.06\epsilon^2.
\end{eqnarray*}
Together with  \eqref{condition}, we get
$${(1-\kappa_{r})\epsilon\over 3} > \sqrt{r}\left\|\calP_\Omega(N_r)-N_r\right\|.$$ 
In light of \eqref{gradient:lowbound}, we have
\begin{eqnarray*}
&&\langle \nabla_X f,\Delta_X\rangle + \langle \nabla_Y f,\Delta_Y\rangle \\
&\geq& \frac{69(1-\kappa_{r})}{100}\epsilon^2-1.5\sqrt{r}\left\| \calP_\Omega(N_r)-N_r \right\|\epsilon \\
&\geq& \left(\frac{69}{100}-{1\over 2}\right)(1-\kappa_{r})\epsilon^2\\
&\geq& {19\over 100}(1-\kappa_{r})\sigma_r\|\Delta_F\|_{\rm F}^2 .
\end{eqnarray*}
The case when \eqref{condition} does not hold can be handled in a similar fashion: since
\begin{align*}
\|D_1\|_{\rm F}&\leq \|D_1+D_2\|_{\rm F} + \|D_2\|_{\rm F} \leq \|D_1+D_2\|_{\rm F} + \|\Delta_F\|_{\rm F}^2 \\
& \leq  \|D_1+D_2\|_{\rm F} + \frac{1-\kappa_{r}}{20}\sqrt{\sigma_r}\|\Delta_F\|_{\rm F} < \frac{63\sqrt{r}}{20(1-\kappa_{r})}\left\|\calP_\Omega(N_r)-N_r\right\|,
\end{align*}
tegother with \eqref{gradient:lowbound}, we have
$$
\langle \nabla_X f,\Delta_X\rangle + \langle \nabla_Y f,\Delta_Y\rangle 
\geq  \frac{69(1-\kappa_{r})}{100}\sigma_r\|\Delta_F\|_{\rm F}^2 - {5r\| \calP_\Omega(N_r)-N_r\|^2\over (1-\kappa_{r})}.
$$
\end{proof}

\begin{proof}[Proof of Lemma \ref{b2}]

By Cauchy-Schwartz inequality, for any $V \in \RR^{d_1\times r}$ such that $\|V\|_\textup{F}=1$,
\begin{eqnarray}
\left|\langle \nabla_X f,V\rangle \right|^2 
&= &\left| \langle \calP_{\Omega}\left(XY^\top - A\right) , VY^\top \rangle +{1\over 2} \langle (X^\top X - Y^\top Y) , X^\top V \rangle\right|^2 \nonumber\\
 &= &2\left| \langle \calP_\Omega(XY^\top-ST^\top)-\calP_\Omega(N_r),VY^\top  \rangle\right|^2  \nonumber\\
 &&\hskip 50pt+ {1\over 2}\left| \langle (X^\top X - Y^\top Y) , X^\top V \rangle\right|^2  \nonumber\\
&\leq &4 \langle \calP_\Omega(XY^\top-ST^\top), VY^\top  \rangle^2 + 4 \langle \calP_\Omega(N_r),VY^\top \rangle^2   \nonumber\\
&&\hskip 50pt +{1\over 2}\left\| X^\top X - Y^\top Y\right\|_{\rm F}^2 \|X\|^2 \|V\|_{\rm F}^2 .\label{d.1}
\end{eqnarray}
The first term of \eqref{d.1} can be bounded by
\begin{eqnarray*}
&&\left| \langle \calP_\Omega(XY^\top-ST^\top),VY^\top\rangle \right|^2 \\
&\leq & \left\langle \calP_\Omega(XY^\top-ST^\top) ,XY^\top-ST^\top \right\rangle \left\langle \calP_\Omega(VY^\top),VY^\top \right\rangle \\
 &\leq& \left[\left\langle \calP_\Omega(S(\Delta_Y)^\top), S(\Delta_Y)^\top \right\rangle^{1/2}+ \left\langle \calP_\Omega((\Delta_X)Y^\top), (\Delta_X)Y^\top \right\rangle^{1/2} \right]^2\cdot\left\langle \calP_\Omega(VY^\top),VY^\top \right\rangle.
\end{eqnarray*}
In the last line, we used the fact that $XY^\top-ST^\top = S(\Delta_Y)^\top + (\Delta_X)Y^\top$ and the Cauchy-Schwartz inequality. We shall need the following concentration inequality.
\begin{lemma}\label{f1}
Assume that each entry of $\Omega$ is independently sampled from binomial trails with probability given by \eqref{eq:defp}. There exists a numerical constant $C_1 >0$ such that, for any matrices $X,Y$, 
\begin{align*}
\langle \calP_{\Omega}(XY^\top),XY^\top\rangle  \leq 2 \|A_r\|_{\rm F}^2 \min \left\{ \|X\|_\textup{F}^2\cdot \max_j  \left(\frac{\|Y_{j.}\|^2}{\|(A_r)_{.j}\|^2}\right),\|Y\|_\textup{F}^2\cdot \max_i  \left(\frac{\|X_{i.}\|^2}{\|(A_r)_{i.}\|^2}\right) \right\},
\end{align*}
with probability at least $1-d^{-\alpha}$ provided that $n \geq C_1\alpha d\log d \cdot \mu_r(A)\|A\|_{\rm F}^2/\|A_r\|_{\rm F}^2$.
\end{lemma}
By Lemma \ref{f1}
\begin{eqnarray*}
&& \left| \langle \calP_\Omega(XY^\top-ST^\top),VY^\top \rangle \right|^2  \\ 
&\leq& \frac{4}{\beta^2(1-\nu_r(A)^2)} \|A_r\|_{\rm F}^2\|\Delta_F\|_\textup{F}^2\cdot \frac{2}{\beta^2(1-\nu_r(A)^2)}  \|A_r\|_{\rm F}^2\|V\|_{\rm F}^2 \\
&=& \frac{8}{\beta^4(1-\nu_r(A)^2)^2} \|A_r\|_{\rm F}^4\|\Delta_F\|_\textup{F}^2,
\end{eqnarray*}
where we used the facts that
$$
\|S_{i.}\| \le \|(A_r)_{i.}\|/\sqrt{\sigma_{r}}\le \|(A_r)_{i.}\|/ (\beta\sqrt{1-\nu_r(A)^2}),$$
and
$$\|Y_{j.}\| \le {\|A_{.j}\|\over\beta} \le  \max_j \left\{{\|A_{.j}\| \over \|(A_r)_{.j}\|} \right\}\cdot {\|(A_r)_{.j}\|\over\beta} = {\|(A_r)_{.j}\|\over \beta\sqrt{1-\nu_r(A)^2}}.$$
Here, we used the definition of $\nu_r(A)$,
$$1-\nu_r(A)^2 = \min_{\substack{1\le i \le d_1\\1\le j \le d_2}}\left\{{\|(A_r)_{i\cdot}\|^2\over \|A_{i\cdot}\|^2},{\|(A_r)_{\cdot j}\|^2\over \|A_{\cdot j}\|^2}\right\} = \left[\max_{\substack{1\le i \le d_1\\1\le j \le d_2}}\left\{{ \|A_{i\cdot}\|^2\over \|(A_r)_{i\cdot}\|^2},{ \|A_{\cdot j}\|^2\over\|(A_r)_{\cdot j}\|^2}\right\}\right]^{-1}.$$
Similarly, we can bound the second term on the righthand side of \eqref{d.1} by
\begin{eqnarray*}
&&\langle \calP_\Omega(N_r), VY^\top\rangle^2 \\
&\leq& 2\langle \calP_\Omega(N_r)-N_r, VY^\top\rangle^2 + 2\langle N_r, VY^\top\rangle^2 \\
&= & 2\langle \calP_\Omega(N_r)-N_r, VY^\top\rangle^2 + 2 \cdot \tr\left( (\Delta_Y)^\top N_r^\top V \right)^2\\
&\leq & 2\| \calP_\Omega(N_r)-N_r\|^2 \cdot \|VY^\top\|_*^2 + 2 \sigma_{r+1}^2 \cdot \|V\|_\textup{F}^2\|\Delta_Y\|_\textup{F}^2\\
&\leq & 2\| \calP_\Omega(N_r)-N_r\|^2 \cdot r\|VY^\top\|_\textup{F}^2 + 2 \sigma_{r+1}^2 \|\Delta_Y\|_\textup{F}^2\\
&\leq & 2r\| \calP_\Omega(N_r)-N_r\|^2 \|Y\|^2 + 2 \sigma_{r+1}^2\|\Delta_Y\|_\textup{F}^2\\
&\leq & 2r\| \calP_\Omega(N_r)-N_r\|^2 \|Y\|^2 + {2 \|A_r\|_{\rm F}^4\over\beta^4(1-\nu_r(A)^2)^2}\|\Delta_Y\|_\textup{F}^2,
\end{eqnarray*}
where we used the fact that 
$${\|A_r\|_{\rm F}^2 \over\beta^2(1-\nu_r(A)^2)}  \geq {\|A_r\|_{\rm F}^2\over \sigma_r} \geq \sigma_r,$$ 
in the last inequality.

The last term of \eqref{d.1} can be bounded by
\begin{eqnarray*}
&&{1\over 2}\left\| X^\top X - Y^\top Y\right\|_{\rm F}^2 \|X\|^2 \|V\|_{\rm F}^2\\
&= &{1\over 2}\left\| X^\top \Delta_X +\Delta_X^\top S - Y^\top \Delta_Y -\Delta_Y ^\top T \right\|_{\rm F}^2 \|X\|^2\\
&\leq & {1\over 2}\left[ (\| X\|+\|S\|) \|\Delta_X\|_{\rm F} + (\| Y\|+\|T\|) \|\Delta_Y\|_{\rm F} \right]^2 \|X\|^2\\ 
&\leq & {9\over 2} \sigma_1^2  \|\Delta_F\|_{\rm F}^2\\
&\leq& {9 \|A_r\|_{\rm F}^4\over 2\beta^4(1-\nu_r(A)^2)^2}\|\Delta_F\|_{\rm F}^2,
\end{eqnarray*}
where we used the facts that  
$$X = S+\Delta_X, \quad Y = T + \Delta_Y,\quad S^\top S = T^\top T,$$
$$\|A_r\|_{\rm F}^2 /(\beta^2(1-\nu_r(A)^2))  \geq \sigma_1,$$
\begin{equation}
\label{X:spectral_bound}
\|X\| \le \|S\| + \|\Delta_X\| \le \|S\| + \|\Delta_F\|_{\rm F} \le \sqrt{\sigma_1} + \frac{1-\kappa_{r}}{20\sqrt{\sigma_1}}\sigma_r \le 1.05\sqrt{\sigma_1},
\end{equation}
and same bound for $\|Y\|$ due to symmetry.

In summary, we have
$$
\| \nabla_X f\|_{\rm F}^2\leq { \|A_r\|_{\rm F}^4\cdot (73\|\Delta_F\|_\textup{F}^2 + 16 \|\Delta_Y\|_\textup{F}^2)\over 2\beta^4(1-\nu_r(A)^2)^2} +9r\sigma_1\| \calP_\Omega(N_r)-N_r\|^2.
$$
The claim then follows from a similar bound for $\| \nabla_Y f\|_{\rm F}$.
\end{proof}

\newpage

\appendix
\section{Proofs of Lemmas}
\begin{proof}[Proof of Lemma \ref{le:concentration}]
Note that
$$
\calP_\Omega(B)-B=\sum_{i,j} {b_{ij}(\omega_{ij}-p_{ij})\over p_{ij}}\cdot e_ie_j^\top=:\sum_{i,j} Z_{ij}.
$$
It is not hard to see that
$$
\left\| \sum_{i,j}\EE Z_{ij}Z_{ij}^\top\right\|=\max_i\left| \sum_{j} {b_{ij}^2(1-p_{ij})\over p_{ij}} \right|  \le {3d\|A\|_{\rm F}^2\over n}\cdot \max_i{\|B_{i\cdot}\|^2\over \|A_{i\cdot}\|^2}.
$$
A similar bound can be derived for $\left\| \sum_{i,j}\EE Z_{ij}^\top Z_{ij}\right\|$.
On the other hand,
$$
\|Z_{ij}\|=\left|{b_{ij}(\omega_{ij}-p_{ij})\over p_{ij}}\right|\le {|b_{ij}|\over p_{ij}}\II(1>p_{ij}>0).
$$
An application of Theorem 4.9 of \cite{latala2018dimension} yields
$$
\PP\left[\|\calP_\Omega(B)-B\|\ge 4\|A\|_{\rm F}\sqrt{d\over n}\max_{\substack{1\le i\le d_1\\1 \le j \le d_2}}\left\{{\|B_{i\cdot}\|\over \|A_{i\cdot}\|},{\|B_{\cdot j}\|\over \|A_{\cdot j}\|}\right\}+t\right]\le d\exp\left(-{C_0t^2\over \max_{i,j:1> p_{ij}>0}\{{b_{ij}^2/ p_{ij}^2}\} }\right).
$$
\end{proof}
\begin{proof}[Proof of Lemma \ref{le:apprest}]
First note that as we have argued in the \eqref{X:spectral_bound}, under the assumptions of Theorem \ref{th:main},
$$\|X_t\| \leq \|S_t\| + \|\Delta_{X_t}\|_{\rm F} < \sqrt{\sigma_1(A)}+0.05\sqrt{\sigma_r(A)} =1.05\sqrt{\sigma_1(A)}.$$
Therefore,
\begin{eqnarray*}
\|X_tY_t^\top-A_r\|_{\rm F}^2&=&\|X_t\Delta_{Y_t}^\top + \Delta_X{T_t}^\top\|_{\rm F}^2  \\
&\le& (\|X_t\|\|\Delta_{Y_t}\|_{\rm F} + \|T_t\|\|\Delta_{X_t}\|_{\rm F})^2\\
&\le& (\|X_t\|^2 + \|T_t\|^2)\|\Delta_{F_t}\|_{\rm F}^2\\
&\le& 3\sigma_1(A)\|\Delta_{F_t}\|_{\rm F}^2,
\end{eqnarray*}
which completes the proof.
\end{proof}

\begin{proof}[Proof of Lemma \ref{c1}]
The first claim is a generalization of Theorem 4.1 from \cite{candes2009exact}, and its proof follows a similar idea. Write $A_r = U_r\Sigma_r V_r^\top$ its SVD and denote by  $X^* = U_r\Sigma_r^{1/2}$ and $Y^* = V_r\Sigma_r^{1/2}$. We first bound $\|\calP_T\calP_\Omega \calP_T-\calP_T\|$, where 
$$T =\{ X^*X^\top + Y{Y^*}^\top |\ (X,Y) \in \RR^{d_2\times r} \times \RR^{ d_1 \times r} \},$$
 and $\calP_T$ is the orthogonal projection on $T$. As shown by \cite{candes2009exact},
$$\calP_T(B) = \calP_{X^*}B+B\calP_{Y^*}-\calP_{X^*}B\calP_{Y^*},$$
for an arbitrary matrix $B$, and
\begin{align}
\left\|\calP_T\left(e_ie_j^\top \right)\right\|_\textup{F}^2
& \leq \left\|\calP_{X^*}e_i\right\|^2  + \left\|\calP_{Y^*}e_j\right\|^2 ,\label{e.2}
\end{align}
where $\calP_{X^*}$, or $\calP_{Y^*}$ is the orthogonal projection onto the column space of $X^*$, or $Y^*$, and $e_i$ is the $i$-th standard basis. 

Observe that
\begin{eqnarray}
\left\|\calP_T\left(A \right)\right\|_\textup{F}^2 &=& \left\|\calP_{X^*} A + (I-\calP_{X^*})A\calP_{Y^*}\right\|_\textup{F}^2 \nonumber\\
&=&\left\|\calP_{X^*} A\right\|_\textup{F}^2 + \left\|(I-\calP_{X^*})A\calP_{Y^*}\right\|_\textup{F}^2 \nonumber\\
& \leq& \left\|A\right\|_\textup{F}^2  + \left\|A\calP_{Y^*}\right\|_\textup{F}^2 \nonumber\\
&\leq& 2\left\|A\right\|_\textup{F}^2.\label{e.1}
\end{eqnarray}
It is easy to verify that $\|(X^*)_{i.}\| \leq \|(A_r)_{i.}\| /\sqrt{\sigma_r} \leq \|A_{i.}\| /\sqrt{\sigma_r}$, then
\begin{eqnarray*}
\|P_{X^*}e_i\|^2 &=& \left\|X^*({X^*}^\top X^*)^{-1}{X^*}^\top e_i\right\|^2\\
 &=&(X^*)_{i.} ({X^*}^\top X^*)^{-1}(X^*)_{i.}^\top \\
& \leq& \frac{\left\|(X^*)_{i.}\right\|^2}{\sigma_{r}} \\
& \leq& \frac{\|A_{i.}\|^2}{\sigma_{r}^2}.
\end{eqnarray*}
We can derive similar result for $\|\calP_{Y^*}e_j\|$. Together with \eqref{e.2}, we have
\begin{align}\label{e.3}
\left\|\calP_T\left(e_ie_j^\top\right)\right\|_\textup{F}^2 \leq \frac{\|A_{i.}\|^2 + \|A_{.j}\|^2}{\sigma_{r}^2}.
\end{align}
For any matrix $B$, we have
$$\calP_\Omega \calP_T(B)-\calP_T(B) = \sum_{ij}\left[\left(\frac{\omega_{ij}}{p_{ij}}-1\right)\langle e_ie_j^\top, \calP_T(B)\rangle e_ie_j^\top\right].$$
It's easy to verify that 
$$\left\langle e_ie_j^\top, \calP_T(B)\right\rangle = \left\langle \calP_T\left(e_ie_j^\top \right), B\right\rangle,$$
by expanding $\calP_T(B)$ and the fact that $\calP_{X^*}$, $\calP_{Y^*}$ are symmetric. 

Therefore
\begin{align*}
\calP_T \calP_\Omega \calP_T(B)-\calP_T(B) &=  \sum_{ij}\left[\left(\frac{\omega_{ij}}{p_{ij}}-1\right)\langle e_ie_j^\top, \calP_T(B)\rangle \calP_T\left(e_ie_j^\top\right)\right] \\
& =  \sum_{ij}\left[\left(\frac{\omega_{ij}}{p_{ij}}-1\right)\left\langle \calP_T\left(e_ie_j^\top \right), B\right\rangle \calP_T\left(e_ie_j^\top\right)\right] \\
& \triangleq \sum_{ij}S_{ij}(B).
\end{align*}

Here, $S_{ij}$ is a transformation from $\mathbb{R}^{d_1\times d_2}$ to $\mathbb{R}^{d_1\times d_2}$, but we can view it as a linear transformation from $\mathbb{R}^{d_1d_2}$ to $\mathbb{R}^{d_1d_2}$ by vectorizing a $d_1\times d_2$ matrix into a $d_1d_2$-dimensional vector. Actually, in the view of $\mathbb{R}^{d_1d_2}$ space,
$$S_{ij} = \left(\frac{\omega_{ij}}{p_{ij}}-1\right) {\rm Vec}\left(P_T\left(e_ie_j^\top \right) \right) {\rm Vec}\left(P_T\left(e_ie_j^\top \right)\right)^\top,$$
where $ {\rm Vec(\cdot)}$ is vectorizing transformation. 

Then we can apply matrix Bernstein inequality.
We have 
$$\mathbbm{E}S_{ij}= 0,\qquad S_{ij}\text{ are mutually independent},$$
 and notice that
$$S_{ij} = 0\qquad\text{if }p_{ij}=1,$$ 
we can always assume
$$ 1 > p_{ij} \geq{n\over 3}\left({\|A_{i\cdot}\|^2\over d_2\|A\|_{\rm F}^2}+{\|A_{\cdot j }\|^2\over d_1\|A\|_{\rm F}^2}+{|a_{ij}|\over \|A\|_{\ell_1}}\right).$$
Together with \eqref{e.3},
\begin{align*}
\|S_{ij}(B)\|_\textup{F} & \leq \frac{1}{p_{ij}} \left\| \calP_T\left(e_ie_j^\top \right) \right\|_\textup{F} \|B\|_\textup{F} \left\|\calP_T\left(e_ie_j^\top\right) \right\|_\textup{F} \\
&\leq \frac{\|A_{i.}\|^2 + \|A_{.j}\|^2}{\sigma_{r}^2p_{ij}} \|B\|_\textup{F} \\
& \leq \frac{3d\|A\|_{\rm F}^2}{n \sigma_{r}^2} \|B\|_\textup{F},
\end{align*}
{therefore}, 
\begin{align}\label{e.4}
\|S_{ij}\| \leq \frac{3d\|A\|_{\rm F}^2}{n \sigma_{r}^2}.
\end{align}
Observe that
\begin{align*}
\left\|\sum_{ij}\mathbbm{E}S_{ij}\left(S_{ij}(B)\right) \right \|_F & =  \left\|\sum_{ij}\mathbbm{E}\left[\left(\frac{\omega_{ij}}{p_{ij}}-1\right)\left\langle e_ie_j^\top, \calP_T(S_{ij}(B))\right\rangle \calP_T\left(e_ie_j^\top\right)\right]\right\|_\textup{F} \\
& = \left\|\sum_{ij}\mathbbm{E}\left[\left(\frac{\omega_{ij}}{p_{ij}}-1\right)^2 \left\langle e_ie_j^\top, \calP_T(B) \right\rangle \left\langle e_ie_j^\top, \calP_T\left(e_ie_j^\top\right)\right\rangle \calP_T\left(e_ie_j^\top\right)\right]\right\|_\textup{F} \\
& = \left\|\sum_{ij}\left[\frac{1-p_{ij}}{p_{ij}} \left\langle e_ie_j^\top, \calP_T(B) \right\rangle \left\langle \calP_T\left(e_ie_j^\top\right), \calP_T\left(e_ie_j^\top\right)\right\rangle \calP_T\left(e_ie_j^\top\right)\right]\right\|_\textup{F} \\
& \leq \sqrt{2}\left\|\sum_{ij}\left[\frac{1-p_{ij}}{p_{ij}} \left\langle e_ie_j^\top, \calP_T(B) \right\rangle \left\langle \calP_T\left(e_ie_j^\top\right), \calP_T\left(e_ie_j^\top\right)\right\rangle e_ie_j^\top \right]\right\|_\textup{F} \\
& \leq \sqrt{2}\left[\max_{ij}\left(\frac{1-p_{ij}}{p_{ij}}\left\|\calP_T\left(e_ie_j^\top \right)\right\|_\textup{F}^2\right)\right] \cdot \left\|\sum_{ij}\left[\left\langle e_ie_j^\top, \calP_T(B)\right\rangle  e_ie_j^\top\right]\right\|_\textup{F} \\
& \leq \frac{3\sqrt{2}d\|A\|_{\rm F}^2}{n \sigma_{r}^2} \left\|\calP_T(B)\right\|_\textup{F}  \\
&\leq \frac{6d\|A\|_{\rm F}^2}{n \sigma_{r}^2} \|B\|_\textup{F},
\end{align*}
where we used \eqref{e.1} for the first and last inequality on the righthand side. Therefore,
$$\left\|\sum_{ij}\mathbbm{E}S_{ij}^2\right\| \leq \frac{6d\|A\|_{\rm F}^2}{n \sigma_{r}^2}.$$
Together with \eqref{e.4}, by matrix Bernstein inequality,
\begin{align}\label{e.5}
\|\calP_T\calP_\Omega \calP_T-\calP_T\| \leq \frac{1-\kappa_r}{20},
\end{align}
with probability at least $1-d^{-\alpha}$, provided that 
$$n \geq C_1(1+\alpha) d\log d \|A\|_{\rm F}^2/(\sigma_r - \sigma_{r+1})^2.$$ 
Note that $D_1 \in T$ by definition,
\begin{eqnarray*}
\langle \calP_\Omega(D_1),D_1\rangle  &=& \langle \calP_\Omega \calP_T(D_1),\calP_T(D_1) \rangle  = \langle \calP_T\calP_\Omega \calP_T(D_1),D_1 \rangle \\
& =&\|D_1\|^2_\textup{F} + \langle  \calP_T\calP_\Omega P_T(D_1) - D_1,D_1\rangle\\
&=&\|D_1\|^2_\textup{F} + \langle  \calP_T\calP_\Omega P_T(D_1) - \calP_T( D_1),D_1\rangle .
\end{eqnarray*}
Together with \eqref{e.5}, we have
\begin{align*}
| \langle \calP_\Omega(D_1),D_1\rangle - \|D_1\|^2_\textup{F} |  \leq \|\calP_T\calP_\Omega \calP_T-\calP_T\| \|D_1\|_\textup{F}^2 \leq \frac{1-\kappa_r}{20}\cdot \|D_1\|_\textup{F}^2.
\end{align*}
We now turn to the second claim. Denote by 
$$\hat{b}_1 =(\|(\Delta_X)_{i.}\|^2/\|(A_r)_{i.}\|)_{1\leq i \leq d_1},\quad \hat{b}_2 =(\|(\Delta_Y)_{j.}\|^2/\|(A_r)_{.j}\|)_{1\leq j \leq d_2},$$ 
two vectors and 
$$B = \{\|(A_r)_{i.}\|\|(A_r)_{.j}\|\}_{1\leq i\leq d_1,1\leq j\leq d_2},$$
a matrix. 
We can assume 
$$p_{ij} \geq{n\over 3}\left({\|A_{i\cdot}\|^2\over d_2\|A\|_{\rm F}^2}+{\|A_{\cdot j }\|^2\over d_1\|A\|_{\rm F}^2}+{|a_{ij}|\over \|A\|_{\ell_1}}\right),$$
as we only care those $p_{ij} < 1$. It's easy to verify that
$$\max_{ij:p_{ij}<1}\left|\frac{b_{ij}}{p_{ij}}\right| \leq \frac{3d\|A\|_{\rm F}^2}{2n}\mu_r^2(A),$$
and
\begin{eqnarray*}
\max_i \sum_{j:p_{ij}<1}\frac{b^2_{ij}}{p_{ij}} &\leq&  \max_i \sum_j\frac{\|(A_r)_{i.}\|^2\|(A_r)_{.j}\|^2}{n \|A_{i.}\|^2/\left(3d\|A\|_{\rm F}^2\right)} \\
& =& \max_i \frac{3d\|A\|_{\rm F}^2\|A_r\|_{\rm F}^2\|(A_r)_{i.}\|^2}{n \cdot \|A_{i.}\|^2} \\
&\leq& \frac{3d\mu_r^2(A)\|A\|_{\rm F}^2\|A_r\|_{\rm F}^2}{n},
\end{eqnarray*}
where we used the fact that $\|(A_r)_{i.}\|\le \mu_r(A)\|A_{i.}\|$ and $\|(A_r)_{.j}\|\le \mu_r(A)\|A_{.j}\|$.
Same bound can be derived for $\max_j \sum_{i:p_{ij}<1}(b^2_{ij}/p_{ij})$. 

Similar to Lemma \ref{le:concentration}, we can prove
\begin{equation*}
\|\calP_{\Omega}(B)-B\| \leq O\left(\sqrt{d/n} \mu_r(A)\|A\|_{\rm F}\|A_r\|_{\rm F}\right),
\end{equation*}
with probability at least $1 - d^{-\alpha}$, provided that $n \geq C_1(1+\alpha)d\log d \cdot  \mu_r^2(A)\|A\|_{\rm F}^2/\|A_r\|_{\rm F}^2$.

Observe that
\begin{eqnarray*}
\left \langle \calP_\Omega(D_2),D_2 \right \rangle &=& \sum_{ij}\frac{\omega_{ij}}{p_{ij}}\langle (\Delta_X)_{i.}, (\Delta_Y)_{j.}\rangle^2\\
&\leq&   \sum_{ij}\frac{\omega_{ij}}{p_{ij}}\|(\Delta_X)_{i.}\|^2\| (\Delta_Y)_{j.}\|^2 \\
&=& \hat{b}_1^\top \left[\calP_\Omega(B)\right] \hat{b}_2\\
&= &\hat{b}_1^\top \left[\calP_\Omega(B)-B\right] \hat{b}_2+ \|\Delta_X\|_\textup{F}^2 \|\Delta_Y\|_\textup{F}^2  \\
 &\leq & \|\calP_{\Omega}(B)-B\| \left\|\hat{b}_1 \right\| \left\|\hat{b}_2\right\|+\|\Delta_X\|_\textup{F}^2 \|\Delta_Y\|_\textup{F}^2\\
&\leq &O\left(\sqrt{d\over n} \mu_r(A)\|A\|_{\rm F} \|A_r\|_{\rm F}\right) \cdot \sqrt{\sum_i \frac{\|(\Delta_X)_{i.}\|^4}{\|(A_r)_{i.}\|^2}\sum_j \frac{\|(\Delta_Y)_{j.}\|^4}{\|(A_r)_{.j}\|^2}} \\
&&\hskip100pt +\|\Delta_X\|_\textup{F}^2 \|\Delta_Y\|_\textup{F}^2 \\
&\le &  O\left((\sqrt{d\over n} \mu_r(A) \|A\|_{\rm F} \|A_r\|_{\rm F}   \right)\cdot\frac{4}{\beta^2(1-\nu_r(A)^2)} \|\Delta_X\|_\textup{F} \|\Delta_Y\|_\textup{F}\\
&&\hskip100pt   +\|\Delta_X\|_\textup{F}^2 \|\Delta_Y\|_\textup{F}^2,
\end{eqnarray*}
where we used the facts that
\begin{eqnarray*}
\|(\Delta_X)_{i.}\|  &\leq& \|X_{i.}\|+\|S_{i.}\| \\
&\leq& \|A_{i.}\|/\beta +  \|(A_r)_{i.}\|/\sqrt{\sigma_{r}}\\
&\leq&    \max_i \left\{{\|A_{i.}\| \over \|(A_r)_{i.}\|} \right\}\cdot {\|(A_r)_{i.}\|\over\beta} +  {\|(A_r)_{i.}\|\over \beta\sqrt{1-\nu_r(A)^2}}\\
&\le& {2\|(A_r)_{i.}\|\over \beta\sqrt{1-\nu_r(A)^2}},
\end{eqnarray*}
and, similarly, $\|(\Delta_Y)_{j.}\|\leq  2 \|(A_r)_{.j}\|/(\beta\sqrt{1-\nu_r(A)^2})$ for the last inequality. 

Together with the fact that
$$\|\Delta_X\|_\textup{F} \|\Delta_Y\|_\textup{F} \leq {1\over2}\|\Delta_F\|_\textup{F}^2\leq  {1\over2}\left(\frac{1-\kappa_r}{20} \right)^2{\sigma_r^2\over \sigma_1} \le  {1\over2}\left(\frac{1-\kappa_r}{20} \right)^2\sigma_r, $$
we get
\begin{equation*}
\left \langle \calP_\Omega(D_2),D_2 \right \rangle \leq \left(\frac{1-\kappa_r}{20}\right)^2\cdot \sigma_r\left\|\Delta_F \right\|_\textup{F}^2,
\end{equation*}
provided
$$n \geq C_1(1+\alpha)d\log d \cdot \frac{\mu_r^2(A) \|A\|_{\rm F}^2\|A_r\|_{\rm F}^2}{(1-\kappa_r)^4\beta^4(1-\nu_r(A)^2)^2\sigma_r^2}.$$ 
\end{proof}

\begin{proof}[Proof of Lemma \ref{f1}]
Let's first look at the case where $x$ and $y$ are vectors,
\begin{align}
\langle \calP_{\Omega}(xy^\top),xy^\top\rangle &= \sum_i x_i^2 \sum_j \frac{\omega_{ij}}{p_{ij}}y_j^2 \nonumber\\
& = \sum_i x_i^2 \sum_j \frac{\omega_{ij}\|(A_r)_{.j}\|^2}{p_{ij}}\cdot \frac{y_j^2}{\|(A_r)_{.j}\|^2}\label{f.1}.
\end{align}
We shall first bound
\begin{align}
&\sum_j \frac{\omega_{ij}\|(A_r)_{.j}\|^2}{p_{ij}} \nonumber\\
= &\sum_{j}\|(A_r)_{.j}\|^2 + \sum_j \frac{(\omega_{ij}-p_{ij})\|(A_r)_{.j}\|^2}{p_{ij}} \nonumber\\
= &\|A_r\|_{\rm F}^2 +  \sum_j \frac{(\omega_{ij}-p_{ij})\|(A_r)_{.j}\|^2}{p_{ij}}\label{f.2}.
\end{align}
Observe that
\begin{align*}
\left| \frac{(\omega_{ij}-p_{ij})\|(A_r)_{.j}\|^2}{p_{ij}} \right| \leq \frac{\|(A_r)_{.j}\|^2}{p_{ij}} \leq \frac{3d\|A\|_{\rm F}^2}{n}\mu_r^2(A),
\end{align*}
and
\begin{eqnarray*}
\mathbbm{E}\sum_j \frac{(\omega_{ij}-p_{ij})^2\|(A_r)_{.j}\|^4}{p_{ij}^2}  &\leq& \sum_j \frac{\|(A_r)_{.j}\|^4}{p_{ij}} \\
&\leq & \sum_j \frac{3d\mu_r^2(A)\|A\|_{\rm F}^2\|(A_r)_{.j}\|^2}{n} \\
&=& \frac{3d\mu_r^2(A)\|A\|_{\rm F}^2\|A_r\|_{\rm F}^2}{n},
\end{eqnarray*}
where we used the facts that $\|(A_r)_{i.}\|\le\mu_r(A)\|A_{i.}\|$ and $\|(A_r)_{.j}\|\le\mu_r(A)\|A_{.j}\|$.

By Bernstein inequality,
\begin{align*}
\sum_j \frac{(\omega_{ij}-p_{ij})\|A_{.j}\|^2}{p_{ij}} \leq O \left(\sqrt{\alpha d\log d/n}\cdot\mu_r(A)\|A\|_{\rm F}\|A_r\|_{\rm F} \right),
\end{align*}
with probability at least $1- d^{-\alpha}$, provided that $n \geq C_1\alpha d\log d\cdot \mu_r^2(A)\|A\|_{\rm F}^2/\|A_r\|_{\rm F}^2$.

Together with \eqref{f.2}, whenever $n \geq C_1\alpha d\log d \cdot \mu_r^2(A)\|A\|_{\rm F}^2/\|A_r\|_{\rm F}^2$,
\begin{align*}
\sum_j \frac{\omega_{ij}\|A_{.j}\|^2}{p_{ij}} \leq 2\|A_r\|_{\rm F}^2.
\end{align*}
Combined with \eqref{f.1},
\begin{align*}
\langle \calP_{\Omega}(xy^\top),xy^\top\rangle & \leq   \sum_i x_i^2 \cdot 2\|A_r\|_{\rm F}^2 \max_j \left( \frac{y_j^2}{\|(A_r)_{.j}\|^2} \right)\\
&= 2\|x\|^2 \|A_r\|_{\rm F}^2\cdot \max_j  \left(\frac{y_j^2}{\|(A_r)_{.j}\|^2}\right).
\end{align*}
Similarly, we have
\begin{align*}
\langle \calP_{\Omega}(xy^\top),xy^\top\rangle  \leq  2\|y\|^2 \|A_r\|_{\rm F}^2\cdot \max_i  \left(\frac{x_i^2}{\|(A_r)_{i.}\|^2}\right).
\end{align*}
Further,
\begin{align*}
\langle \calP_{\Omega}(XY^\top),XY^\top\rangle  = \sum_{i,j} \frac{\omega_{ij}}{p_{ij}}\langle X_{i.},Y_{j.}\rangle^2 \leq \sum_{i,j} \frac{\omega_{ij}}{p_{ij}}\| X_{i.}\|^2\|Y_{j.}\|^2 .
\end{align*}
Applying the previous results to vector $x = (\| X_{i.}\|)_{1 \leq i \leq d_1}$ and $y = (\| Y_{j.}\|)_{1 \leq j \leq d_2}$ gives us the result in the matrix form.
\end{proof}

\bibliographystyle{plainnat}
\bibliography{reference}

\end{document}